\documentclass[reqno,oneside,11pt]{amsart}

\usepackage{amssymb,graphicx,braket,amsmath,amsfonts,amsthm}
\usepackage{hyperref}
\usepackage[utf8]{inputenc}
\usepackage{braket}
\usepackage{color}

\theoremstyle{plain}
\newtheorem{theorem}{Theorem}[section]
\theoremstyle{plain}
\newtheorem{proposition}[theorem]{Proposition}
\newtheorem*{proposition*}{Proposition}
\newtheorem{remark}[theorem]{Remark}
\newtheorem*{theorem*}{Theorem}

\newtheorem{lemma}[theorem]{Lemma}
\newtheorem{assumption}[theorem]{Assumption}
\newtheorem*{assumption*}{Assumption}
\newtheorem{corollary}[theorem]{Corollary}

\theoremstyle{definition} 
\newtheorem{definition}[theorem]{Definition}

\numberwithin{equation}{section}

\def\mainmatter{\def\baselinestretch{1.1}\normalfont}

\usepackage[margin=1in]{geometry}

\makeatletter
\renewcommand{\section}{\@startsection
{section}%                   % the name
{1}%                         % the level
{\z@}%                       % the indent / 0mm
{-\baselineskip}%            % the before skip / -3.5ex \@plus -1ex \@minus -.2ex
{0.8\baselineskip}%          % the after skip / 2.3ex \@plus .2ex
{\centering\scshape\large}} % the style

\renewcommand{\subsection}{\@startsection
{subsection}%                   % the name
{2}%                         % the level
{\z@}%                       % the indent / 0mm
{-0.8\baselineskip}%            % the before skip / -3.5ex \@plus -1ex \@minus -.2ex
{0.5\baselineskip}%          % the after skip / 2.3ex \@plus .2ex
{\normalfont \bf \normalsize}} % the style

\renewcommand{\subsubsection}{\@startsection
{subsubsection}%                   % the name
{3}%                         % the level
{\z@}%                       % the indent / 0mm
{-0.8\baselineskip}%            % the before skip / -3.5ex \@plus -1ex \@minus -.2ex
{0.5\baselineskip}%          % the after skip / 2.3ex \@plus .2ex
{\normalfont \it \normalsize}} % the style
\makeatother

\newcommand{\eprint}[1]{\href{http://arxiv.org/abs/#1}{\texttt{arXiv\string:\allowbreak#1}}}

\begin{document}

\title{Topological Recursion in the Ramond Sector}

\author{Kento Osuga}
\address{Theoretical Physics Institute, Department of Physics,
University of Alberta, 4-181 CCIS,
Edmonton, Alberta, T6G 2E1, Canada}
\address{School of Mathematics and Statistics,
University of Sheffield,
The Hicks Building,
Hounsfield Road,
Sheffield, S3 7RH,
United Kingdom}
\email{osuga@ualberta.ca}

\begin{abstract}

We investigate supereigenvalue models in the Ramond sector and their recursive structure. We prove that the free energy truncates at quadratic order in Grassmann coupling constants, and consider super loop equations of the models with the assumption that the $1/N$ expansion makes sense. Subject to this assumption, we obtain the associated genus-zero algebraic curve with two ramification points (one regular and the other irregular) and also the supersymmetric partner polynomial equation. Starting with these polynomial equations, we present a recursive formalism that computes all the correlation functions of these models. Somewhat surprisingly, correlation functions obtained from the new recursion formalism have no poles at the irregular ramification point due to a supersymmetric correction -- the new recursion may lead us to a further development of supersymmetric generalizations of the Eynard-Orantin topological recursion.

\end{abstract}

\maketitle
\tableofcontents
\mainmatter

\newpage
\section{Introduction}\label{sec:intro}

Hermitian matrix models are in some sense the simplest gauge quantum field theories, namely in zero dimensions. Despite their simplicity, however, they give rise to a variety of rich applications in both physics and mathematics. What particularly interesting is a recursive formalism called topological recursion which was originally introduced in \cite{CEO} as a technique that solves loop equations for Hermitian matrix models. Shortly after, it was significantly generalized by \cite{Alexandrov:2006qx,EO,EO2} beyond matrix models and it has been applied in various contexts in mathematical physics such as Gromov-Witten invariants, Hurwitz numbers, knot invariants in Chern-Simons theory, and more (see \cite{BKMP,BM,BEM,DBOSS,EMS,EO3,FLZ2,FLZ3,GJKS,Ma} and references therein). Such an abstract recursive formalism has become known as the so-called Eynard-Orantin topological recursion\footnote{Even though some methods of solving the loop equations were known before the formulation of Eynard and Orantin, we refer to the universal recursive formalism that not only solves the loop equations but works beyond matrix models as the \emph{Eynard-Orantin topological recursion}.}. 

Another important property of Hermitian matrix models is that the partition function obeys the so-called Virasoro constraint. From this perspective, we can consider super-generalized models, known as supereigenvalue models, such that their partition function satisfies the super Virasoro constraint (see \cite{AP,AlvarezGaume:1991jd,AlvarezGaume:1992mm,Beckers,BO,C1, C2,Kroll,McArthur,Plefka1,Plefka2,Plefka:1996tt} for the Neveu-Schwarz (NS) sector and \cite{C} for the Ramond sector). Then one may ask whether a similar recursive structure also appears in these models. And furthermore if so, can such a recursive formalism lead us to interesting applications in mathematics?

For the NS sector, \cite{BO} showed that all correlation functions can be recursively computed by the Eynard-Orantin topological recursion in conjunction with an auxiliary Grassmann-valued polynomial equation, which can be thought of as the supersymmetric partner of the algebraic curve for the Eynard-Orantin topological recursion. This is because of great simplification thanks to Becker's formula \cite{Beckers,McArthur} that gives an explicit relation to Hermitian matrix models. At the same time, supereigenvalue models in the NS sector are so highly constrained by Becker's formula that there are only few hints to consider a supersymmetric generalization of topological recursion. As an extension of their work, in this paper we investigate supereigenvalue models in the Ramond sector introduced in \cite{C} and uncover their recursive structure.

In fact, we find several interesting results in supereigenvalue models in the Ramond sector. First of all, we prove that the free energy for the Ramond sector depends only on Grassmann couplings up to quadratic order, similar to that for the NS sector. However, they do not seem to relate to Hermitian matrix models, which causes trouble with justifying the $1/N$ expansion of the partition function and the free energy. Therefore, we rather \emph{assume} that the $1/N$ expansion makes sense in those models, and derive their super loop equations. By analyzing the super loop equations, we obtain an algebraic curve of genus zero where one of the ramification point is regular (Airy-like) and the other is irregular (Bessel-like), and we also derive the supersymmetric partner which is a Grassmann-valued polynomial equation defined on the ordinal algebraic curve (not on a super Riemann surface). Starting with the two polynomial equations, we present a new recursive formalism that computes all the correlation functions of supereigenvalue models in the Ramond sector where some of corretaion functions are Grassmann-valued. Supersymmetric corrections contribute to higher genus correlation functions, and most importantly, no correlation functions have a pole at the irregular ramification point due to a supersymmetric correction. Hence, the recursive formalism evidently differs from the Eynard-Orantin topological recursion unlike what is shown in \cite{BO} for the NS sector. We summarize the results of this paper in Theorem~\ref{thm:main}.

We are working in a more abstract framework that potentially unifies both the NS sector and Ramond sector into one package. However, since such abstraction is beyond the scope of this paper, we leave more discussion to a paper in progress \cite{BO2}. We also note that it is an open question whether the new recursive formalism shown in this paper (or perhaps the more abstract one in \cite{BO2}) has interesting applications beyond the study of supereigenvalue models.

This paper is organized as follows. In Section 2, we briefly review a definition of supereigenvalue models in the Ramond sector, and study the properties of the partition function and the free energy. With the assumption of the $1/N$ expansion, the bosonic and fermionic loop equations are derived in Section 3. Then in Section 4, we focus on the recursive structure of the loop equations independent of Grassmann couplings, and present a recursive formalism that can be thought of as a generalization of the Eynard-Orantin topological recursion. Section 5 is devoted to correlation functions that depend on Grassmann couplings, and also to the role of a supersymmetric partner polynomial equation. In Section 6, we outline a variety of open questions and future work, particularly a recursive formalism potentially unifying both the NS sector and Ramond sector.

\subsection*{Acknowledgements}

We would like to thank Vincent Bouchard for insightful discussions. We also acknowledge Nitin Chidambaram, Don N. Page, and Piotr Su\l kowski for helpful advice. This research was supported by the Natural Sciences and Engineering Research Council of Canada, and also by the Engineering and Physical Sciences Research
Council under grant agreement ref. EP/S003657/1.

\section{Supereigenvalue Models in the Ramond Sector}\label{sec:SEM}

In this section, we define a \emph{formal} supereigenvalue model in the Ramond sector. A general construction of such models are given in \cite{C} which considers two distinct types of supereigenvalue models in the Ramond sector. We are only interested in what in \cite{C} is called the Ramond-NS eigenvalue model, and leave the other type, the Ramond-Ramond eigenvalue model, for future work. Also, we restrict our focus on models without $\alpha/\beta$-deformation. Even such restricted models give rise to interesting results as outlined in Section~\ref{sec:intro}. Discussions and computations in the paper closely follow the presentation of \cite{BO,Eynard}.

\subsection{Definitions}
For nonnegative integers $k,l\in\mathbb{Z}_{\geq0}$, let $g_k,\xi_l$ be formal bosonic and fermionic (Grassmann) coupling constants. Then, we define power series $V(x),\Psi(x)$, which we call the bosonic and fermionic potential respectively, by
\begin{equation}
V(x)=Tx+\sum_{k\geq0}g_kx^k,\;\;\;\;\Psi(x)=\sum_{l\geq0}\xi_lx^l,\label{potential}
\end{equation}
where $T\in\mathbb{R}_{>0}$ is a new parameter. Then, the partition function $\mathcal{Z}$ of a \emph{formal supereigenvalue model in the Ramond sector} is defined as\footnote{This integral representation \eqref{Z1} is slightly different from the original one given in \cite{C}. We will explain why we take this representation shortly.}
\begin{equation}
\mathcal{Z}=\prod_{k,l\geq0}\sum_{n_k,m_l\geq0}\int d\lambda d\theta \Delta(\lambda,\theta)\frac{1}{n_k!m_l!}\left(-\frac{N}{t}\sum_{i=1}^{2N}g_k\lambda_i^{2k}\right)^{n_k}\left(-\frac{N}{t}\sum_{i=1}^{2N}\xi_l\theta_i\lambda_i^{2l}\right)^{m_l}\;e^{-\frac{NT}{2t}\sum_{i=1}^{2N}\lambda_i^2},\label{Z1}
\end{equation}
\begin{equation}
d\lambda=\prod_{i=1}^{2N} d\lambda_i,\;\;\;\;d\theta=\prod_{i=1}^{2N} d\theta_i,
\end{equation}
where $N\in\mathbb{N}$, $t\in\mathbb{R}_{>0}$ are also new parameters, and $\Delta(\lambda,\theta)$ will be determined shortly. We integrate every $\lambda_i$ over $\mathbb{R}$, and $\theta_i$ are Grassmann variables. $\mathcal{Z}$ depends on $N,t,g_k,\xi_l,T$, but we omit to explicitly denote those dependence for simplicity. In general summation and integral in \eqref{Z1} do not commute. For convention, however, we often denote the partition function by
\begin{equation}
\mathcal{Z}\overset{\text{formal}}{=}\int d\lambda d\theta \Delta(\lambda,\theta) e^{-\frac{N}{t}\sum_{i=1}^{2N}(V(\lambda_i^2)+\Psi(\lambda_i^2)\theta_i)},
\end{equation}
with the understanding that the summation is taken outside of the integral as originally defined in \eqref{Z1}.

Note that we will eventually require the bosonic and fermionic potentials $V(x),\Psi(x)$ to be polynomials of some degrees when we investigate a recursive structure of the models. However, it is more convenient to leave them as power series for now.

\subsubsection{Super Virasoro Constraints}

We define super Virasoro differential operators $L_n,G_m$ for $n,m\in\mathbb{Z}_{\geq0}$ in terms of $(g_k,\xi_l)$ that generate a super Virasoro subalgebra in the Ramond sector:
\begin{align}
L_n=&T\frac{\partial}{\partial g_{n+1}}+\sum_{k\geq0}kg_k\frac{\partial}{\partial g_{k+n}}+\frac{1}{2}\left(\frac{t}{N}\right)^2\sum_{k=0}^n\frac{\partial}{\partial g_k}\frac{\partial}{\partial g_{n-j}}+\frac{1}{16}\delta_{n,0}\nonumber\\
&+\sum_{k\geq0}\left(\frac{n}{2}+k\right)\xi_{k}\frac{\partial}{\partial\xi_{n+k}}+\frac{1}{2}\left(\frac{t}{N}\right)^2\sum_{k=0}^{n}\left(\frac{n}{2}-k\right)\frac{\partial}{\partial\xi_{k}}\frac{\partial}{\partial\xi_{n-k}}-\frac{n}{4}\left(\frac{t}{N}\right)^2\frac{\partial}{\partial\xi_0}\frac{\partial}{\partial\xi_{n}},\label{L}\\
G_n=&T\frac{\partial}{\partial \xi_{n+1}}+\sum_{k\geq0}\left(kg_k\frac{\partial}{\partial\xi_{n+k}}+\xi_{k}\frac{\partial}{\partial g_{n+k}}\right)-\left(\frac{t}{2N}\right)^2\frac{\partial}{\partial \xi_0}\frac{\partial}{\partial g_n}+\sum_{k=0}^{n}\left(\frac{t}{N}\right)^2\frac{\partial}{\partial\xi_k}\frac{\partial}{\partial g_{n-k}},\label{G}
\end{align}
\begin{align}
[L_m, L_n]&=(m-n)L_{m+n}, \nonumber \\
[L_m, G_r]&=\frac{m-2r}{2}G_{m+r} \label{algebra}, \\
\{G_r, G_s\}&=2L_{r+s}-\frac{1}{8}\delta_{r+s,0}.\nonumber
\end{align}
Note that we can obtain differential operators for the full super Virasoro algebra, but we only need the subalgebra to define a formal supereigenvalue model in the Ramond sector.

Similar to supereigenvalue models in the NS sector, we would like to impose the partition function \eqref{Z1} to be annihilated by the super Virasoro operators \eqref{L} and \eqref{G}. In other words, the partition function can be viewed as a highest weight state in the Ramond sector:
\begin{align}
G_n\mathcal{Z}&=0,\\
L_n\mathcal{Z}&=\delta_{n,0}\frac{1}{16}\mathcal{Z},
\end{align}
which is the so-called \emph{super Virasoro constraint}. Note that the second equality is automatically satisfied as a consequence of $G_n\mathcal{Z}=0$ because of the super Virasoro subalgebra. Plugging the definition of the partition function \eqref{Z1} into the super Virasoro constraint, it can be shown that $\Delta(\lambda,\theta)$ is uniquely determined, up to an overall normalization, as
\begin{equation}
\Delta(\lambda,\theta)=\prod_{i<j}^{2N}\left(\lambda_i^2-\lambda_j^2-\frac{\theta_i\theta_j}{2}(\lambda_i^2+\lambda_j^2)\right).\label{D}
\end{equation}
\begin{proof}
We provide two approaches. The first one is simply by re-parametrization of integration variables of the partition function given in Section 4, \cite{C}. Explicitly, \cite{C} showed that the partition function is given in the form:
\begin{equation}
\mathcal{Z}\overset{\text{formal}}{=}\int d\lambda d\theta \Delta(\lambda,\theta) \exp\left(-\frac{N}{t}\sum_{i=1}^{2N}(V(\lambda_i)+\Psi(\lambda_i)\frac{\theta_i}{\lambda_i^{\frac12}})\right),
\end{equation}
\begin{equation}
\Delta(\lambda,\theta)=\prod_{i<j}^{2N}\left(\lambda_i-\lambda_j-\frac{\lambda_i+\lambda_j}{2\lambda_i^{\frac12}\lambda_j^{\frac12}}\theta_i\theta_j\right),
\end{equation}
where all $\lambda_i$'s are nonnegative real variables. If we re-parametrize them as follows
\begin{equation}
\lambda_i\rightarrow\lambda_i^2,\;\;\;\;\lambda_i^{-\frac{1}{2}}\theta_i\rightarrow\theta_i,
\end{equation}
then the integrand depends only on even powers of $\lambda_i$. Thus, we can extend $\lambda_i$ to be in $\mathbb{R}$. This derives \eqref{D}.

We can directly prove \eqref{D} too. The super Virasoro constraint $G_n\mathcal{Z}=0$ induces a set of differential equations that $\Delta(\lambda,\theta)$ should obey:
\begin{equation}
\forall n\in\mathbb{Z}_{\geq0},\;\;\;\;\sum_{i=1}^{2N}\lambda_i^{2n}\left(\frac{\lambda_i\theta_i}{2}\frac{\partial}{\partial \lambda_i}-\frac{\partial}{\partial\theta_i}\right)\Delta(\lambda,\theta)=\Delta(\lambda,\theta)\sum_{i\neq j}\theta_i\frac{2\lambda_i^{2n+2}-\lambda_i^2\lambda_j^{2n}-\lambda_j^{2n+2}}{2(\lambda_i^2-\lambda_j^2)}.
\end{equation}
Then, one can show by induction in $2N\geq2$ that \eqref{D} is a unique solution up to an overall constant similar to the proof for supereigenvalue models in the NS sector.
\end{proof}

\begin{remark}
In contrast to the definition given in \cite{C}, our definition \eqref{Z1} is a Gaussian integral representation, hence, several computational techniques developed for Hermitian matrix models can be easily carried on. Therefore, we will proceed with the definition \eqref{Z1}.
\end{remark}

\subsection{Truncation}\label{sec:truncation}
From this section and below, we set $T=1$. Let us define the free energy of a formal supereigenvalue model in the Ramond sector as
\begin{equation}
\mathcal{F}=\log\mathcal{Z}.\label{F1}
\end{equation}
\cite{McArthur} proved (see also Appendix A of \cite{BO}) that the free energy of those in the NS sector depends on Grassmann couplings $\xi^{NS}_{l+\frac12}$ only up to quadratic order. We now prove that this feature also holds in the Ramond sector:

\begin{proposition}\label{prop:trun}
The free energy $\mathcal{F}$ depends on Grassmann couplings $\xi_l$ only up to quadratic order:
\begin{equation}
\mathcal{F}=\mathcal{F}^{(0)}+\mathcal{F}^{(2)},
\end{equation}
where the superscript denotes the order of $\xi_l$-dependence.
\end{proposition}
\begin{proof}
The proof closely follows the analysis shown in \cite{McArthur} which involves careful permutation and re-parametrization of $\lambda_i,\theta_i$. In terms of Grassmann couplings $\xi_l$, we can expand the partition function in the form
\begin{equation}
\mathcal{Z}=\sum_{K\geq0}\mathcal{Z}^{(2K)},
\end{equation}
where the superscript $(2K)$ denotes the order of $\xi_l$-dependence. Note that the possible highest order is $2N$ no matter what the degree of the fermionic potential $\Psi(x)$ is. Recall that Grassmann integrals obey
\begin{equation}
\int d\theta_k=0,\;\;\;\;\int \prod_{i=1}^{2N}d\theta_i\theta_{\sigma(1)}\cdots\theta_{\sigma(2N)}=\text{sgn}(\sigma),\label{Grassmann integral}
\end{equation}
where $\sigma \in S_{2N}$. Thus, we can express $\mathcal{Z}^{(2K)}$ as
\begin{align}
\mathcal{Z}^{(2K)}\overset{\text{formal}}{=}&\left(\frac{N}{t}\right)^{2K}\int d\lambda\prod_{i<j}^{2N}(\lambda_i^2-\lambda_j^2)e^{-\frac{N}{t}\sum_{i=1}^{2N}V(\lambda_i^2)}\nonumber\\
&\hspace{20mm}\times\left(\frac{1}{2K!}\int d\theta\prod_{i<j}^{2N}\left(1-\frac{\theta_i\theta_j}{2}\frac{\lambda_i^2+\lambda_l^2}{\lambda_i^2-\lambda_j^2}\right)\left(\sum_{i=1}^{2N}\Phi(\lambda^2_i)\theta_i\right)^{2K}\right).\label{2K1}
\end{align}

The second line of \eqref{2K1}, i.e., Grassmann integrals, can be computed as
\begin{align}
&\frac{1}{2K!}\int d\theta\prod_{i<j}^{2N}\left(1-\frac{\theta_i\theta_j}{2}\frac{\lambda_i^2+\lambda_l^2}{\lambda_i^2-\lambda_j^2}\right)\left(\sum_{i=1}^{2N}\Phi(\lambda^2_i)\theta_i\right)^{2K}\nonumber\\
&=\frac{1}{(2K)!2^{N-K}(N-K)!}\left(\frac{N}{t}\right)^{2K}\sum_{\sigma\in S_{2N}}(-1)^{\sigma}\prod_{i=1}^{2K}\Phi(\lambda^2_{\sigma(i)})\prod_{j=K+1}^N\frac{1}{2}\frac{\lambda^2_{\sigma(2j)}+\lambda^2_{\sigma(2j-1)}}{\lambda^2_{\sigma(2j)}-\lambda^2_{\sigma(2j-1)}}.
\label{2K2}
\end{align}
Also, the Vandermonde-like factor in the first line of \eqref{2K1} is expanded as
\begin{equation}
\prod_{i<j}^{2N}(\lambda^2_i-\lambda^2_j)=(-1)^N\sum_{\rho\in S_{2N}}(-1)^{\rho}\prod_{l=1}^{2N}\lambda_l^{2(\rho(l)-1)}.
\end{equation}
Thus, by plugging these into \eqref{2K1}, we have
\begin{align}
\mathcal{Z}^{(2K)}\overset{\text{formal}}{=}&\frac{(-1)^N}{(2K)!2^{N-K}(N-K)!}\left(\frac{N}{t}\right)^{2K}\int d\lambda e^{-\frac{N}{t}\sum_{i=1}^{2N}V(\lambda^2_i)}\sum_{\rho,\sigma\in S_{2N}}(-1)^{\rho+\sigma}\nonumber\\
&\;\;\;\;\times\prod_{l=1}^{2N}\lambda_l^{2(\rho(l)-1)}\prod_{i=1}^{2K}\Phi(\lambda^2_{\sigma(i)})\prod_{j=K+1}^N\frac{1}{2}\frac{\lambda^2_{\sigma(2j)}+\lambda^2_{\sigma(2j-1)}}{\lambda^2_{\sigma(2j)}-\lambda^2_{\sigma(2j-1)}}.
\end{align}

We can further simplify the expression. Notice that since every $\lambda_i$ is integrated, we can freely rename $\lambda_{\sigma(i)}\rightarrow\lambda_i$. As a result, the summation over $\sigma$ simply gives $(2N)!$ and $\mathcal{Z}^{(2K)}$ becomes
\begin{align}
\mathcal{Z}^{(2K)}\overset{\text{formal}}{=}&\frac{(-1)^N(2N)!}{(2K)!2^{N-K}(N-K)!}\left(\frac{N}{t}\right)^{2K}\sum_{\rho\in S_{2N}}(-1)^{\rho}\prod_{i=1}^{2K}\int d\lambda_i e^{-\frac{N}{t}V(\lambda^2_i)}\lambda_i^{2(\rho(i)-1)}\Psi(\lambda^2_i)
\nonumber\\
&\;\;\;\;\times\prod_{j=K+1}^N\int d\lambda_{2j-1}d\lambda_{2j} e^{-\frac{N}{t}(V(\lambda^2_{2j-1})+V(\lambda^2_{2j}))}\lambda_{2j-1}^{2(\rho(2j-1)-1)}\lambda_{2j}^{2(\rho(2j)-1)}\frac{1}{2}\frac{\lambda^2_{2j}+\lambda^2_{2j-1}}{\lambda^2_{2j}-\lambda^2_{2j-1}}.\nonumber\\
\end{align}
If we define an antisymmetric matrix $A$ and a Grassmann-valued vector $\zeta$ as
\begin{align}
A_{ij}\overset{\text{formal}}{=}&\int d\lambda d\sigma e^{-\frac{N}{t}(V(\lambda^2)+V(\sigma^2))}\lambda^{2(i-1)}\sigma^{2(j-1)}\frac{1}{2}\frac{\lambda^2+\sigma^2}{\lambda^2-\sigma^2},\\
\zeta_i\overset{\text{formal}}{=}&\frac{N}{t}\int d\lambda e^{-\frac{N}{t}V(\lambda^2)}\lambda^{2(i-1)}\Psi(\lambda^2).
\end{align}
Then, we arrive at
\begin{equation}
Z^{(2K)}=\frac{(-1)^N(2N)!}{(2K)!2^{N-K}(N-K)!}\sum_{\lambda\in S_{2N}}(-1)^{\lambda}\prod_{i=1}^{2K}\zeta_{\lambda(i)}\prod_{j=K+1}^NA_{\lambda(2j)\lambda(2j-1)}.\label{2K3}
\end{equation}

If the partition function can be decomposed into this form, it is proven that by using the properties of Pfaffians, the free energy is written in the form:
\begin{equation}
\mathcal{F}=\log\mathcal{Z}^{(0)}+\frac{1}{2}\zeta^TA^{-1}\zeta.\label{2K4}
\end{equation}
We refer to \cite{McArthur}, or Appendix A in \cite{BO}, to fulfil the steps from \eqref{2K3} to \eqref{2K4}.
\end{proof}

\begin{remark}
This is an interesting observation. The free energies in NS sector and the Ramond sector both truncate at quadratic order in Grassmann couplings. It remains to be seen whether we can conceptually understand why this feature is true in both sectors.
\end{remark}

\begin{remark}
An analogous relation of Becker's formula \cite{Beckers} has not been discovered yet. In Section~\ref{sec:recursion}, however, we derive a recursive formalism that computes every correlation function of the model even without such a powerful simplification.
\end{remark}

\subsection{Large $N$ Expansion}\label{sec:Large N}

The large $N$ expansion is a common technique for matrix models, and it provides us a perturbative expansion order by order in $1/N$. It originates from 't Hooft's argument for ribbon graphs, each of which we can assign the Euler characteristics. For Hermitian matrix models, one can show more properties in terms of t'Hooft's parameter $t$ and coupling constants $g^H_k$ for $k\geq3$ that become crucial to derive the Eynard-Orantin topological recursion. We refer the readers to \cite{Eynard} and references therein for more details, and we simply list those properties below:

\begin{itemize}
\item  Let $Z^H,F^H$ be the partition function and the free energy of a formal $N\times N$ 1-Hermitian matrix model, then they possess the $1/N$ expansion:
\begin{equation}
Z^H=\sum_{g\geq0}\left(\frac{N}{t}\right)^{2-2g}Z^H_g,\;\;\;\;F^H=\sum_{g\geq0}\left(\frac{N}{t}\right)^{2-2g}F^H_g,
\end{equation}
where $Z^H_g,F^H_g$ are independent of $N$,
\item For a fixed $g$, all $Z^H_g,F^H_g$ are formal power series in $t$ (if we set $g^H_0=g^H_1=g^H_2=0$),
\item Order by order in $t$, $Z^H_g$ and $F^H_g$ are polynomials in $g^H_k$ for $k\geq3$.\end{itemize}

It seems difficult in general to show that supereigenvalue models also have these properties because their partition functions are not written in terms of matrix integrals. Thanks to Becker's formula \cite{Beckers} and McArthur's truncation formula \cite{McArthur}, however, it can be proven that the partition function and the free energy in the NS sector still enjoy these properties since they are related to Hermitian matrix models.  An analogous relation for the Ramond sector is still under investigation, hence we do not have rigorous evidence whether the $1/N$ expansion makes sense in the Ramond sector. In this section, therefore, we will see how the $t$-dependence and $(g_k,\xi_l)$-dependence appear in the Ramond sector by directly looking at Gaussian integral representations \eqref{Z1}.

\subsubsection{$t$-Dependence and $(g_k,\xi_l)$-Dependence}

Let us investigate the $t$-dependence and $(g_k,\xi_l)$-dependence of the partition function $\mathcal{Z}$ and the free energy $\mathcal{F}$. To start with, recall that $\Delta(\lambda,\theta)$ \eqref{D} is uniquely fixed by the super Virasoro constraint up to an overall normalization which we denote by $\mathcal{N}$. We now fix the normalization $\mathcal{N}$ by the following condition:
\begin{equation}
\mathcal{N}\mathcal{Z}\bigr|_{g_k=\xi_l=0}=1.\label{N}
\end{equation}
After integrating all Grassmann variables $\theta_i$ in the partition function, we find
\begin{align}
\mathcal{N}\mathcal{Z}\bigr|_{g_k=\xi_l=0}\propto&\;\;\mathcal{N}\left(\prod_{i=1}^{2N}\int d\lambda_i\lambda_i^{2n_i}\exp\left(-\frac{N}{2t}\lambda_i^2\right)\right),\label{NZ}\\
\sum_{i=1}^{2N}n_i=&\;\text{number of possible pairings}=N(2N-1),
\end{align}
where the proportionality \eqref{NZ} holds up to some $N$-dependence. Then we operate the Gaussian integrals in terms of $\lambda_i$ to get
\begin{equation}
\mathcal{N}\mathcal{Z}_{g_k=\xi_l=0}\propto\;\mathcal{N}\prod_{i=1}^{2N}\left(\frac{t}{N}\right)^{n_i+\frac{1}{2}}=\mathcal{N}\left(\frac{t}{N}\right)^{2N^2}.
\end{equation}
This determines that the $t$-dependence of $\mathcal{N}$ is given by $\mathcal{N}\propto t^{-2N^2}$. From now on, whenever we refer to the partition function, we mean the one normalized by the condition \eqref{N}.

We next study the $g_0$-dependence. It turns out that it is very simple because the summation and the integral commute for $g_0$, and we can write the partition function in the form
\begin{equation}
\mathcal{Z}=e^{\frac{-2N^2}{t}g_0}\hat{\mathcal{Z}},\label{g_0}
\end{equation}
where $\hat{\mathcal{Z}}$ does not depend on $g_0$. The free energy is then given by
\begin{equation}
\mathcal{F}=N^2\frac{2g_0}t+\log\hat{\mathcal{Z}}.
\end{equation}
Thus, $\mathcal{F}$ linearly depends on $g_0$ and nowhere else. As one can see, the partition function and the free energy cannot be formal series in $t$ unless $g_0=0$. Note that this feature also appears for Hermitian matrix models; the partition function for Hermitian matrix models becomes a power series in $t$ and $g^H_k$ after we set $g^H_0=g^H_1=g^H_2=0$. See \cite{Eynard} for more details.

We now turn to investigate the $g_k$-dependence of $\hat{\mathcal{Z}}^{(0)}$ for $k\geq1$. In this case, the Gaussian integrals also pick contributions from nonzero couplings $g_k$ as
\begin{align}
\hat{\mathcal{Z}}^{(0)}\sim&\;\;\mathcal{N}\prod_{k\geq1}^d\sum_{m_k\geq0}\left(g_k\frac{N}{t}\right)^{m_k}\int\prod_{i=1}^{2N} d\lambda_i\lambda_i^{2n_i}\sum_{i_k=1}^{2N}\lambda_{i_k}^{2km_k}\exp\left(-\frac{N}{2t}\sum_{j=1}^{2N}\lambda_j^2\right)\nonumber\\
\sim&\prod_{k\geq1}\sum_{m_k\geq0}\left(\frac{t}{N}\right)^{(k-1)m_k}g_k^{m_k},\label{Z2}
\end{align}
where $\sim$ only shows how $\hat{\mathcal{Z}}^{(0)}$ depends on $t$ and $g_k$, but \eqref{Z2} does not tell the $N$-dependence. This proves that $\hat{\mathcal{Z}}^{(0)}$ is a power series in $t$. Also, for all $k\geq2$, \eqref{Z2} implies that order by order in $t$, $\hat{\mathcal{Z}}^{(0)}$ is a polynomial in $g_{k\geq2}$ which resembles to one of the key features of Hermitian matrix models. The case for $m=1$ is different; there can be power series of $g_1$ at any order in $t$ unlike any other $g_{k\geq2}$\footnote{This is also not surprising. The $g_1$-dependence of the partition function for the Ramond sector is analogous to the $g_2^H$-dependence for Hermitian matrix models.}. To put it another way, once we set $g_0=g_1=0$, the partition function $\hat{\mathcal{Z}}^{(0)}$ is a power series in $t$, and it is a polynomial in $g_{k\geq2}$ order by order in $t$.

Similar arguments can be applied for the $\xi_l$-dependent terms $\mathcal{Z}^{(2K)}$ for $1\leq K\leq N$. The only irregular factor at $t=0$ appears in $\mathcal{Z}^{(2)}$ in the form
\begin{equation}
\mathcal{Z}^{(2)}\sim\frac{\xi_0\xi_1}{t}\hat{f}[[g_1]](N),
\end{equation}
where $\hat{f}[[g_1]](N)$ is a power series in $g_1$ with some unknown $N$-dependence. By counting the degree of the Gaussian integral similar to \eqref{Z2}, it is straightforward to see that all other terms in $\mathcal{Z}^{(2K)}$ are power series of $t$, and order by order in $t$ they are polynomials in $\xi_{k\geq0},g_{k\geq2}$. Thus, we can summarize the $(t,g_k,\xi_l)$-dependence of $\mathcal{Z}$ as
\begin{equation}
\mathcal{Z}=e^{\frac{-2N^2}{t}g_0}\left(\frac{\xi_0\xi_1}{t}\hat{f}[[g_1]](N)+\sum_{v\geq0}t^v\tilde{\mathcal{Z}}_v[[g_1]][g_k,\xi_l](N)\right),\label{Zt}
\end{equation}
where $\tilde{\mathcal{Z}}_v[[g_1]][g_k,\xi_l](N)$ are power series in $g_1$ but polynomials in $g_{k\geq2}$ and $\xi_l$, and we do not know their $N$-dependence. For $v=0, g_1=0$ they satisfy $\tilde{\mathcal{Z}}_v[[g_1]][g_k,\xi_l](N)=1$ so that it is consistent with the normalization condition \eqref{N}

Knowing the expansion \eqref{Zt} and the normalization condition \eqref{N}, the free energy $\mathcal{F}$ can be given by expanding $\log\mathcal{Z}$ in terms of $\xi_l$, $g_1$, and $t$. In summary, we can write it as
\begin{equation}
\mathcal{F}=-N^2\frac{2g_0}{t}+\frac{\xi_0\xi_1}{t}\hat{h}[[g_1]](N)+\sum_{v\geq0}t^v\tilde{\mathcal{F}}_v[[g_1]][g_k,\xi_l](N),\label{Ft}
\end{equation}
where $\hat{h}[[g_1]](N)$ are power series in $g_1$, and $\tilde{\mathcal{F}}_v[[g_1]][g_k,\xi_l](N)$ are power series in $g_1$ but polynomials in $g_{k\geq2}$ and $\xi_l$. Their $N$-dependence will be discussed shortly. For $v=0, g_1=0$ they satisfy $\tilde{\mathcal{F}}_v[[g_1]][g_k,\xi_l](N)=0$. Note that due to Proposition~\ref{prop:trun}, $\tilde{\mathcal{F}}_v[g_k,\xi_l](N)$ truncate at quadratic order in $\xi_l$, even though it is not evident in the expression \eqref{Ft}.

\subsubsection{Large $N$ Assumption}

We have discussed the $t$-dependence and the $(g_k,\xi_l)$-dependence of the partition function and the free energy, but their $N$-dependence is more complex. This is because the $N$-dependence would appear not only by Gaussian integrals, but also by nontrivial summation in $\Delta(\lambda,\theta)$, summation from coupling terms, and permutation of Grassmann integrals. Meanwhile, the definition of supereigenvalue models in the Ramond sector seems a natural analogue of that in the NS sector which indeed enjoys the $1/N$ expansion. Therefore from now on, we rather \emph{assume} that the $1/N$ expansion still holds in the Ramond sector, and explore their recursive structure with the assumption:

\begin{assumption}\label{assumption}
The partition function $\mathcal{Z}$ \eqref{Z1} and the free energy $\mathcal{F}$ \eqref{F1} posses the $1/N$ expansion. That is, we have
\begin{equation}
\mathcal{Z}=\bar{\mathcal{Z}}\sum_{g\geq0}\left(\frac{N}{t}\right)^{2-2g}\mathcal{Z}_g,\;\;\;\;\mathcal{F}^{(0,2)}=\log\bar{\mathcal{Z}}+\sum_{g\geq0}\left(\frac{N}{t}\right)^{2-2g}\mathcal{F}_g^{(0,2)},
\end{equation}
where $\mathcal{Z}_g,\mathcal{F}_g^{(0,2)}$ are independent of $N$, and $\bar{\mathcal{Z}}$ depends only on $N$.
\end{assumption}

\begin{remark}
Since this is strictly speaking the $t/N$ expansion, $\mathcal{Z}_g,\mathcal{F}_g$ would potentially have negative powers of $t$, at most of degree $1-2g$. This does not create any significant issue in computation, though it is important to keep in mind that it is lower bounded for a fixed $g$. Note that $\mathcal{F}_0$ is a genuine power series in $t$.
\end{remark}

\begin{remark}
Every proposition and theorem shown below is subject to Assumption~\ref{assumption}. It is under investigation whether the $1/N$ expansion can be proven with mathematical rigour. We still note that \eqref{Zt}, \eqref{Ft}, and Proposition~\ref{prop:trun} hold regardless.
\end{remark}

\section{Loop Equations}\label{sec:Loop Equation}

In this section we shall define correlation functions and derive loop equations for supereigenvalue models in the Ramond sector. See \cite{BO,C1,C2} and references therein for analogous arguments in the NS sector. \cite{C} also discuss loop equations in the Ramond sector in the context of super quantum curves.
 
\subsection{Correlation Functions}
For a supereigenvalue model in the Ramond sector, an \emph{expectation value} $\braket{E(\lambda,\theta)}$ of a function $E(\lambda,\theta)$ of $\lambda_i$ and $\theta_i$ is defined as
\begin{equation}
\braket{E(\lambda,\theta)}\overset{\text{formal}}{=}\frac{1}{\mathcal{Z}}\int d\lambda d\theta \Delta(\lambda,\theta)E(\lambda,\theta) e^{-\frac{N}{t}\sum_{i=1}^{2N}(V(\lambda_i^2)+\Psi(\lambda_i^2)\theta_i)}.
\end{equation}
Recall that the integrand of the partition function only depends on even powers of $\lambda_i$, hence any odd power dependence of $\lambda$ in $E(\lambda,\theta)$ vanishes. Also, since the integrand is symmetric in terms of permutation of $\lambda_i$ and anti-symmetric in terms of $\theta_i$, we impose such symmetry on $E(\lambda,\theta)$ too. Therefore, quantities of our interest are the following:
\begin{equation}
\mathcal{T}_{k_1\cdots k_n|l_1 \cdots l_m}=\sum_{a_1,\cdots,a_n=1}^{2N}\sum_{b_1,\cdots,b_m=1}^{2N}\braket{\lambda_{a_1}^{2k_1}\cdots\lambda_{a_n}^{2k_n}\theta_{b_1}\lambda_{b_1}^{2l_1}\cdots\theta_{b_n}\lambda_{b_n}^{2l_n}},\label{T}
\end{equation}
where $k_i,l_j\in\mathbb{Z}_{\geq0}$. Similar to Hermitian matrix models and supereigenvalue models in the NS sector, it turns out that it is more convenient to put them into one package as a set of generating functions of connected parts of \eqref{T}. However, we need to utilize a trick to define these generating functions that are recursively computable.

Let us define \emph{bosonic} and \emph{fermionic loop insertion operators} by
\begin{align}
\frac{\partial}{\partial V(x)}&=-\sum_{k\geq0}\frac{1}{x^{k+1}}\frac{\partial}{\partial g_k},\\\frac{\partial}{\partial \Psi(X)}&=-\sum_{l\geq0}\frac{1}{X^{l+1}}\frac{\partial}{\partial \xi_l}+\frac{1}{2X}\frac{\partial}{\partial \xi_0},\label{Finsertion}
\end{align}
where we will explain shortly why we inserted the last term in \eqref{Finsertion}. Then, we define such generating functions of connected parts of $\mathcal{T}_{k_1\cdots k_n|l_1 \cdots l_m}$ by acting an arbitrary number of times with the bosonic and fermionic loop insertion operators on $\mathcal{F}$
\begin{equation}
\mathcal{W}_{n|m}(J|K)=\left(\frac{t}{N}\right)^{n+m}\prod_{i=1}^n\prod_{j=1}^m\frac{\partial}{\partial V(x_i)}\frac{\partial}{\partial \Psi(X_j)}\mathcal{F},
\end{equation}
where we introduced the notation $J=(x_1,...,x_n)$ and $K=(X_1,...,X_m)$. We give a few remarks here:

\begin{remark}\label{rem:cor}
\hfill
\begin{itemize}
\item  Suppose all $\mathcal{W}_{n|m}(J|K)$ are given, then, $\mathcal{T}^c_{k_1\cdots k_n|l_1 \cdots l_m}$ appear as coefficients of $\mathcal{W}_{n|m}(J|K)$ after expanding them in terms of $1/x_i,1/X_j$ where we denote by the superscript $c$ connected parts of $\mathcal{T}_{k_1\cdots k_n|l_1 \cdots l_m}$. If one of $l_m$ were zero, we would need to multiply the coefficient by $2$ to get the right $\mathcal{T}^c_{k_1\cdots k_n|l_1 \cdots l_m}$ due to the last term in the fermionic insertion operator \eqref{Finsertion}.

\hfill
\item In CFT, the Laurent expansion of the free fermion in the Ramond sector is normally give in half-integers powers $1/X^{n+\frac12}$. (See \cite{C}.) However, in this paper we take the definition \eqref{Finsertion} as a power series of $1/X$ because with we do not need to consider the ambiguity of $X^{\frac12}$ when we consider the correlation functions as meromorphic differentials on an algebraic curve. We will discuss in \cite{BO2} a more abstract framework which takes $X^{\frac12}$ into account.

\hfill
\item  The last term in the fermionic loop insertion operator \eqref{Finsertion} can be also explained from a CFT perspective. Namely, the fermionic zero mode $\psi_0$ obeys $\{\psi_0,\psi_0\}=1$. This implies that the zero mode $\psi_0$ can be represented by a Grassmann coupling $\xi_0$ as
\begin{equation}
\psi_0=\xi_0+\frac{1}{2}\frac{\partial}{\partial\xi_0}.
\end{equation}
The $\frac12$ factor in the second term explains why the last term in \eqref{Finsertion} is needed. Equivalently, if we define the loop insertion operator without the last term in \eqref{Finsertion}, then we should define the fermionic potential $\Psi(x)$ \eqref{potential} with the factor of $1/2$ in front of $\xi_0$. It turns out that without this adjustment of the $\frac12$ factor, the loop equations are not written by a combination of the correlation functions and some polynomials of $x,X$ any more.
\end{itemize}
\end{remark}

\subsubsection{The Large $N$ Expansion of Correlation Functions}

Let us present an important consequence of Assumption~\ref{assumption}:

\begin{lemma}\label{lemma:poly}
For $g\in\mathbb{Z}_{\geq0}$, $t^{1-2g}\cdot\mathcal{F}_g^{(0,2)}$ are power series in $g_1$ and $t$, and order by order in $t$, they are polynomials in $g_{k\geq2},\xi_{l\geq0}$.\end{lemma}
\begin{proof}
This is straightforward from \eqref{Zt} and \eqref{Ft} with Assumption~\ref{assumption}.
\end{proof}

Assumption~\ref{assumption} allows us to expand $\mathcal{W}_{n|m}(J|K)$ in terms of power series in $1/N$. Namely, we define $\mathcal{W}_{g,n|m}$ by the $1/N$ expansion of $\mathcal{W}_{n|m}$ as follows
\begin{align}
\mathcal{W}_{n|m}(J|K)=&\sum_{g\geq0}\left(\frac{N}{t}\right)^{2-2g-n-m}\mathcal{W}_{g,n|m},\\
\mathcal{W}_{g,n|m}(J|K)=&\left(\frac{t}{N}\right)^{n+m}\prod_{i=1}^n\prod_{j=1}^m\frac{\partial}{\partial V(x_i)}\frac{\partial}{\partial \Psi(X_j)}\mathcal{F}_g.\label{defW}
\end{align}
Furthermore, Proposition~\ref{prop:trun} implies that we can classify $\mathcal{W}_{g,n|m}$ into four types in terms of the $\xi_l$-dependence:
\begin{itemize}
\item Pure bosonic type $\mathcal{W}^{(0)}_{g,n|0}(J|)$,
\item Linear Grassman type $\mathcal{W}^{(1)}_{g,n|1}(J|X)$,
\item Quadratic Grassman type $\mathcal{W}^{(2)}_{g,n|0}(J|)$,
\item Two-fermion type $\mathcal{W}^{(0)}_{g,n|2}(J|X_1,X_2)$.
\end{itemize}
Note that in our notation, two-fermion type correlation functions are independent of Grassmann couplings; two-fermion means that we act with two fermionic loop insertion operators on the free energy. Finally, notice that whenever we act with a bosonic or fermionic loop insertion operator on the free energy, it effectively replaces couplings $g_k$ or $\xi_l$ with $x^{-k-1}$ or $X^{-l-1}$. Therefore, Lemma~\ref{lemma:poly} implies the following corollary
\begin{corollary}\label{coro:poly}
For $n,m,g\in\mathbb{Z}_{\geq0}$ with $n+m\geq1$, all $\mathcal{W}_{g,n|m}(J|K)$ are lower-bounded Laurent formal series expansion in $t$ whose lowest degree is at most $t^{1-2g+n+m}$. Order by order in $t$, $\mathcal{W}_{g,n|m}(J|K)$ are polynomials in $1/x,1/X$.
\end{corollary}

\subsubsection{Linear and Quadratic Grassmann Types}

We note that one can compute $\mathcal{W}^{(1)}_{g,n|1}(J|X)$ and $\mathcal{W}^{(2)}_{g,n|0}(J|)$ if we know all $\mathcal{W}^{(0)}_{g,n|2}(J|X_1,X_2)$. In fact, if we expand $\mathcal{W}^{(1)}_{g,n|1}(J|X)$ as
\begin{equation}
\mathcal{W}^{(1)}_{g,n|1}(J|X)=\sum_{l\geq0}\xi_l\mathcal{W}^{(1,l)}_{g,n|1}(J|X),
\end{equation}
then we recover the coefficients by
\begin{equation}
\mathcal{W}^{(1,k)}_{g,n|1}(J|X)=(1+\delta_{k,0})\underset{\tilde{X}\rightarrow\infty}{{\rm Res}}\tilde{X}^k\mathcal{W}^{(0)}_{g,n|2}(J|\tilde{X},X)d\tilde{X}.\label{W1formula0}
\end{equation}
This is an immediate consequence by the definition of the fermionic loop insertion operators where the $\delta_{k,0}$ factor is the adjustment due to the last term in \eqref{Floop}. Or equivalently, one can write \eqref{W1formula0} with the fermionic potential $\Psi(X)$ as
\begin{equation}
\underset{\tilde{X}\rightarrow\infty}{{\rm Res}}\left(\Psi(\tilde{X})+\xi_0\right)\frac{\partial}{\partial\Psi(\tilde{X})}=\sum_{l\geq0}\xi_l\frac{\partial}{\partial\xi_l},\label{I0}
\end{equation}
\begin{equation}
\mathcal{W}^{(1)}_{g,n|1}(J|X)=\underset{\tilde{X}\rightarrow\infty}{{\rm Res}}\left(\Psi(\tilde{X})+\xi_0\right)\mathcal{W}^{(0)}_{g,n|2}(J|\tilde{X},X)d\tilde{X}.\label{W1formula}
\end{equation}
Note that \eqref{I0} is proportional to the identity operator for $\xi_l$-dependent correlation functions. Similarly, $\mathcal{W}^{(2)}_{g,n|0}(J|)$ can be obtained by
\begin{equation}
\mathcal{W}^{(2)}_{g,n|0}(J|)=\frac12\underset{X\rightarrow\infty}{{\rm Res}}\left(\Psi(X)+\xi_0\right)\mathcal{W}^{(1)}_{g,n|2}(J|X)dX.\label{W2formula}
\end{equation}

Therefore, every quantity of our interest \eqref{T} in the model can be computed as long as we have all $\xi_l$-independent correlation functions $\mathcal{W}^{(0)}_{g,n|0}(J|)$ and $\mathcal{W}^{(0)}_{g,n|2}(J|X_1,X_2)$, and $\Psi(X)$.

\subsection{Super Loop Equations}

Our aim for the rest of the paper is to explore recursive relations among these correlation functions. The starting point is to derive the loop equations in the Ramond sector, and we try to uncover the recursive structure from them. Recall that for the NS sector, the bosonic and fermionic loop equations are derived based on the following two equations:
\begin{equation}
\sum_{n\geq-1}\frac{1}{x^{n+2}}L^{NS}_n\mathcal{Z}_{NS}=0,\;\;\;\;\sum_{r\geq-\frac12}\frac{1}{X^{r+\frac32}}G^{NS}_r\mathcal{Z}_{NS}=0
\end{equation}
After manipulating these equations in terms of the bosonic and fermionic loop insertion operators, we arrive at the bosonic and fermionic loop equations for the NS sector. See Appendix B in \cite{BO} for the detailed derivation.

\subsubsection{Bosonic Loop Equation}
For the Ramond sector, we are missing $L_{-1}\mathcal{Z}=0$ unlike the NS sector. Also, the partition function is not annihilated by $L_0$ any more, but rather it is shifted by the term proportional to the central charge. Thus, \cite{C} derived the analogous loop equation for the Ramond sector from the following equation\footnote{\cite{C} showed that this equation is equivalent to the invariance of the partition function under transformations of integration variables}:
\begin{equation}
\sum_{n\geq0}\frac{1}{x^{n+2}}L_n\mathcal{Z}=\frac{\mathcal{Z}}{16x^2}.
\end{equation}
After some manipulation, we can rewrite the above equation in terms of correlation functions which we call the \emph{bosonic loop equation}:
\begin{align}
0=&-\frac{N}{t}V'(x)\mathcal{W}^{(0,2)}_{1|0}(x|)+\frac{1}{2}\mathcal{W}^{(0,2)}_{2|0}(x,x|)+\frac{1}{2}\mathcal{W}^{(0,2)}_{1|0}(x|)^2+\mathcal{P}^{(0,2)}_{1|0}(x|)\nonumber\\
&+\frac{N}{2t}\left(\left(\frac{\Psi(z)}{x}-\Psi'(x)\right)\mathcal{W}^{(1)}_{0|1}(|x)+\Psi(x)\frac{d}{dx}\mathcal{W}^{(1)}_{0|1}(|x)\right)+\mathcal{Q}^{(2)}_{1|0}(x|)\nonumber\\
&-\frac{1}{2}x\mathcal{W}^{(1)}_{0|1}(|x)\frac{d}{dx}\mathcal{W}^{(1)}_{0|1}(|x)-\frac{1}{2}x\frac{d}{d\tilde{x}}\mathcal{W}^{(0)}_{0|2}(|x,\tilde{x})\biggr|_{x=\tilde{x}},\nonumber\\\label{Bloop}
\end{align}
where
\begin{align}
\mathcal{P}^{(0,2)}_{1|0}(x|)=&-\sum_{m\geq0}x^{m-1}\sum_{k\geq0}(m+k+1)(\delta_{m+k,0}+g_{m+k+1})\frac{\partial}{\partial g_k}\mathcal{F}^{(0,2)},\label{P}\\
\mathcal{Q}^{(2)}_{1|0}(x|)=&-\frac12\sum_{l\geq0}x^{l-1}\left(\sum_{m\geq0}\frac{l+2m+1}{1+\delta_{m,0}}\xi_{l+m+1}\frac{\partial}{\partial \xi_m}\right)\mathcal{F}^{(2)}.\label{Q}
\end{align}
See Appendix~\ref{sec:derivation} for the derivation. The $\delta_{m+k,0}$ factor in \eqref{P} is the remnant of $T=1$ whereas the $\delta_{k,0}$ factor in \eqref{Q} is a consequence of the fermionic loop insertion operator \eqref{Finsertion}. Note that \eqref{Bloop} has $\xi_l$-independent terms, terms quadratic in $\xi_l$, and a quartic one in $\xi_l$ where the quartic one simply gives $\mathcal{W}^{(2)}_{1|0}(x|)^2=0$. This is analogous to the one in the NS sector and a consequence of the quadratic truncation of the free energy.

Furthermore, by acting an arbitrary number of times with the bosonic loop insertion operator on the bosonic loop equation, we obtain the \emph{general bosonic loop equation}

\begin{align}
0=&-\frac{N}{t}V'(x)\mathcal{W}^{(0,2)}_{n+1|0}(x,J|)+\frac{1}{2}\mathcal{W}^{(0,2)}_{n+2|0}(x,x,J|)+\frac{1}{2}\sum_{J_1\cup J_2=J}\mathcal{W}^{(0,2)}_{n_1+1|0}(x,J_1|)\mathcal{W}^{(0,2)}_{n_2+1|0}(x,J_2|)\nonumber\\
&+\sum_{i=1}^n\frac{d}{dx_i}\frac{\mathcal{W}^{(0,2)}_{n|0}(x,J\backslash x_i|)-\mathcal{W}^{(0,2)}_{n|0}(J|)}{x-x_i}+\frac{1}{x}\sum_{i=1}^n\frac{d}{dx_i}\mathcal{W}^{(0,2)}_n(J|)+\mathcal{P}^{(0,2)}_{n+1|0}(x,J|)\nonumber\\
&+\frac{N}{2t}\left(\left(\frac{\Psi(z)}{x}-\Psi'(x)\right)\mathcal{W}^{(1)}_{n|1}(J|x)+\Psi(x)\frac{d}{dx}\mathcal{W}^{(1)}_{n|1}(J|x)\right)+\mathcal{Q}^{(2)}_{n+1|0}(x,J|)\nonumber\\
&-\frac{1}{2}x\sum_{J_1\cup J_2=J}\mathcal{W}^{(1)}_{n_1|1}(J_1|x)\frac{d}{dx}\mathcal{W}^{(1)}_{n_2|1}(J_2|x)-\frac{1}{2}x\frac{d}{d\tilde{x}}\mathcal{W}^{(0)}_{n|2}(J|x,\tilde{x})\biggr|_{x=\tilde{x}},\nonumber\\\label{GBloop}
\end{align}
where
\begin{align}
\mathcal{P}^{(0,2)}_{n+1|0}(x,J|)=&-\left(\frac{t}{N}\right)^n\sum_{m\geq0}x^{m-1}\sum_{k\geq0}(m+k+1)(\delta_{m+k,0}+g_{m+k+1})\frac{\partial}{\partial g_k}\prod_{i=1}^n\frac{\partial}{\partial V(x_i)}\mathcal{F}^{(0,2)}.\\
\mathcal{Q}^{(2)}_{n+1|0}(x,J|)=&-\frac12\left(\frac{t}{N}\right)^n\sum_{l\geq0}x^{l-1}\left(\sum_{m\geq0}\frac{l+2m+1}{1+\delta_{m,0}}\xi_{l+m+1}\frac{\partial}{\partial \xi_m}\right)\prod_{i=1}^n\frac{\partial}{\partial V(x_i)}\mathcal{F}^{(2)}.
\end{align}
and we used the following identities in the derivation:
\begin{align}
\frac{\partial}{\partial V(x_{n+1})}V'(x)=&\frac{d}{dx_{n+1}}\frac{1}{x_1-x},\\
\frac{t}{N}\frac{\partial}{\partial V(x_{n+1})}\mathcal{P}^{(0,2)}_{n+1|0}(x,J|)=&-\frac{d}{dx_{n+1}}\frac{x_{n+1}}{x}\frac{\mathcal{W}^{(0,2)}_{n+1|0}(J,x_{n+1}|)}{x-x_{n+1}}+\mathcal{P}^{(0,2)}_{n+2|0}(x,J,x_{n+1}|)\nonumber\\
=&\frac{1}{x}\frac{d}{dx_{n+1}}\mathcal{W}^{(0,2)}_{n+1|0}(J,x_{n+1}|)-\frac{d}{dx_{n+1}}\frac{\mathcal{W}^{(0,2)}_{n+1|0}(J,x_{n+1}|)}{x-x_{n+1}}+\mathcal{P}^{(0,2)}_{n+2|0}(x,J,x_{n+1}|).
\end{align}

Finally, if we consider the $1/N$ expansion of \eqref{GBloop}, the coefficients of $(N/t)^{2-2g-n}$ gives the following set of equations:
\begin{align}
0=&-V'(x)\mathcal{W}^{(0,2)}_{g,n+1|0}(x,J|)+\frac{1}{2}\mathcal{W}^{(0,2)}_{g-1,n+2|0}(x,x,J|)+\frac{1}{2}\sum_{\substack{g_1+g_2=g\\ J_1\cup J_2=J}}\mathcal{W}^{(0,2)}_{g_1,n_1+1|0}(x,J_1|)\mathcal{W}^{(0,2)}_{g_2,n_2+1|0}(x,J_2|)\nonumber\\
&+\sum_{i=1}^n\frac{d}{dx_i}\frac{\mathcal{W}^{(0,2)}_{g,n+1|0}(x,J\backslash x_i|)-\mathcal{W}^{(0,2)}_{g,n+1|0}(J|)}{x-x_i}+\frac{1}{x}\sum_{i=1}^n\frac{d}{dx_i}\mathcal{W}^{(0,2)}_{g,n}(J|)+\mathcal{P}^{(0,2)}_{g,n+1|0}(x,J|)\nonumber\\
&+\frac{1}{2}\left(\left(\frac{\Psi(z)}{x}-\Psi'(x)\right)\mathcal{W}^{(1)}_{g,n|1}(J|x)+\Psi(x)\frac{d}{dx}\mathcal{W}^{(1)}_{g,n|1}(J|x)\right)+\mathcal{Q}^{(2)}_{g,n+1|0}(x,J|)\nonumber\\
&-\frac{1}{2}x\sum_{\substack{g_1+g_2=g\\ J_1\cup J_2=J}}\mathcal{W}^{(1)}_{g_1,n_1|1}(J_1|x)\frac{d}{dx}\mathcal{W}^{(1)}_{g_2,n_2|1}(J_2|x)-\frac{1}{2}x\frac{d}{d\tilde{x}}\mathcal{W}^{(0)}_{g-1,n|2}(J|x,\tilde{x})\biggr|_{x=\tilde{x}},\nonumber\\\label{BWgn1}
\end{align}
where
\begin{align}
\mathcal{P}^{(0,2)}_{g,n+1|0}(x,J|)=&-\left(\frac{t}{N}\right)^n\sum_{m\geq0}x^{m-1}\sum_{k\geq0}(m+k+1)(\delta_{m+k,0}+g_{m+k+1})\frac{\partial}{\partial g_k}\prod_{i=1}^n\frac{\partial}{\partial V(x_i)}\mathcal{F}_g^{(0,2)}.\\
\mathcal{Q}^{(2)}_{g,n+1|0}(x,J|)=&-\frac12\left(\frac{t}{N}\right)^n\sum_{l\geq0}x^{l-1}\left(\sum_{m\geq0}\frac{l+2m+1}{1+\delta_{m,0}}\xi_{l+m+1}\frac{\partial}{\partial \xi_m}\right)\prod_{i=1}^n\frac{\partial}{\partial V(x_i)}\mathcal{F}_g^{(2)}.
\end{align}

Hereafter, we refer to \eqref{BWgn1} the bosonic loop equations rather than \eqref{Bloop} for simplicity.

\subsubsection{Fermionic Loop Equation}
For the fermionic loop equation, \cite{C} showed that the fermionic loop equation can be derived from the following half-integer power series 
\begin{equation}
\sum_{m\geq0}\frac{1}{X^{m+\frac32}}G_m\mathcal{Z}=0.\label{Floop1}
\end{equation}
However, as mentioned in Remark~\ref{rem:cor}, $X^{\frac12}$ does not get along well with known techniques developed for the Eynard-Orantin topological recursion. Therefore, we instead consider an equation induced by
\begin{equation}
\sum_{m\geq0}\frac{1}{X^{m+1}}G_m\mathcal{Z}=0.
\end{equation}
That is, we simply multiply by $X^{\frac12}$ \eqref{Floop1} to make it a power series expansion in $1/X$. As shown in Appendix~\ref{sec:derivation}, this power series yields:
\begin{align}
0=&-\frac{N}{t}\left(\Psi(X)\mathcal{W}_{1|0}^{(0,2)}(X|)+XV'(X)\mathcal{W}_{0|1}^{(1)}(|X)\right)\nonumber\\
&+X\Bigl(\mathcal{W}_{1|0}^{(0,2)}(X|)\mathcal{W}_{0|1}^{(1)}(|X)+\mathcal{W}_{1|1}^{(1)}(X|X)\Bigr)+\mathcal{R}_{0|1}^{(1,3)}(|X),\label{Floop}
\end{align}
where
\begin{equation}
\mathcal{R}_{0|1}^{(1,3)}(|X)=-\sum_{m,l\geq0}X^l\left(\xi_{l+m+1}\frac{\partial}{\partial g_m}+\frac{(l+m+1)}{1+\delta_{m,0}}(\delta_{l+m,0}+g_{l+m+1})\frac{\partial}{\partial \xi_m}\right)\mathcal{F}^{(0,2)}
\end{equation}

Similarly to the general bosonic loop equations, by acting an arbitrary number of times with the bosonic loop insertion operator on the fermionic loop equation, we obtain the \emph{general fermionic loop equations}
\begin{align}
0=&-\frac{N}{t}\left(\Psi(X)\mathcal{W}_{n+1|0}^{(0,2)}(X,J|)+XV'(X)\mathcal{W}_{n|1}^{(1)}(J|X)\right)\nonumber\\
&+\sum_{i=1}^n\frac{d}{dx_i}\frac{X\mathcal{W}_{n-1|1}^{(1)}(J\backslash x_i|X)-x_{n+1}\mathcal{W}_{n-1|1}^{(1)}(J\backslash x_i|x_i)}{X-x_i}\nonumber\\
&+X\Bigl(\sum_{J_1\cup J_2=J}\mathcal{W}_{n_1+1|0}^{(0,2)}(X,J_1|)\mathcal{W}_{n_2|1}^{(1)}(J_2|X)+\mathcal{W}_{n+1|1}^{(1)}(X,J|X)\Bigr)+\mathcal{R}_{n|1}^{(1,3)}(J|X),\label{GFloop}
\end{align}
where
\begin{align}
\mathcal{R}_{n|1}^{(1,3)}(J|X)=&-\left(\frac{t}{N}\right)^n\sum_{m,l\geq0}X^l\left(\xi_{l+m+1}\frac{\partial}{\partial g_m}+\frac{(l+m+1)}{1+\delta_{m,0}}(\delta_{l+m,0}+g_{l+m+1})\frac{\partial}{\partial \xi_{m}}\right)\nonumber\\
&\hspace{30mm}\times\prod_{i=1}^n\frac{\partial}{\partial V(x_i)}\mathcal{F}^{(0,2)},
\end{align}
and we used
\begin{equation}
\frac{t}{N}\frac{\partial}{\partial V(x_{n+1})}\mathcal{R}_{n|1}^{(1,3)}(J|X)=-\frac{d}{dx_{n+1}}\frac{x_{n+1}\mathcal{W}_{n|1}^{(1)}(J|x_{n+1})}{X-x_{n+1}}+\mathcal{R}_{n+1|1}^{(1,3)}(J,x_{n+1}|X).
\end{equation}
The $1/N$ expansion then gives
\begin{align}
0=&-\Psi(X)\mathcal{W}_{g,n+1|0}^{(0,2)}(X,J|)-XV'(X)\mathcal{W}_{g,n|1}^{(1)}(J|X)\nonumber\\
&+\sum_{i=1}^n\frac{d}{dx_i}\frac{X\mathcal{W}_{g,n-1|1}^{(1)}(J\backslash x_i|X)-x_{n+1}\mathcal{W}_{g,n-1|1}^{(1)}(J\backslash x_i|x_i)}{X-x_i}\nonumber\\
&+X\Bigl(\sum_{\substack{g_1+g_2=g\\J_1\cup J_2=J}}\mathcal{W}_{g_1,n_1+1|0}^{(0,2)}(X,J_1|)\mathcal{W}_{g_2,n_2|1}^{(1)}(J_2|X)+\mathcal{W}_{g-1,n+1|1}^{(1)}(X,J|X)\Bigr)+\mathcal{R}_{g,n|1}^{(1,3)}(J|X),\label{FWgn1}
\end{align}
where
\begin{align}
\mathcal{R}_{g,n|1}^{(1,3)}(J|X)&=-\left(\frac{t}{N}\right)^n\sum_{m,l\geq0}X^l\left(\xi_{l+m+1}\frac{\partial}{\partial g_{m}}+\frac{(l+m+1)}{1+\delta_{m,0}}(\delta_{l+m,0}+g_{l+m+1})\frac{\partial}{\partial \xi_{m}}\right)\nonumber\\
&\hspace{30mm}\times\prod_{i=1}^n\frac{\partial}{\partial V(x_i)}\mathcal{F}_g^{(0,2)}
\end{align}

Note that if we act with the fermionic loop insertion operator on the general fermionic loop equation \eqref{GFloop} and expand them in terms of both $1/N$ and $\xi_l$, then the $\xi_l$-independent coefficients order by order in $1/N$ gives the loop equation for $\mathcal{W}_{g,n|2}^{(0)}(J|X,\tilde{X})$
\begin{align}
0=&-XV'(X)\mathcal{W}_{g,n|2}^{(0)}(J|X,\tilde{X})+\frac{\mathcal{W}_{g,n+1|0}^{(0)}(X,J|)-\mathcal{W}_{g,n+1|0}^{(0)}(\tilde{X},J|)}{X-\tilde{X}}+\frac{\mathcal{W}_{g,n+1|0}^{(0)}(X,J|)}{2\tilde{X}}\nonumber\\
&+\sum_{i=1}^n\frac{d}{dx_i}\frac{X\mathcal{W}_{g,n-1|2}^{(0)}(J\backslash x_i|X,\tilde{X})-x_{n+1}\mathcal{W}_{g,n-1|2}^{(0)}(J\backslash x_i|x_i,\tilde{X})}{X-x_i}\nonumber\\
&+X\Bigl(\sum_{\substack{g_1+g_2=g\\J_1\cup J_2=J}}\mathcal{W}_{g_1,n_1+1|0}^{(0)}(X,J_1|)\mathcal{W}_{g_2,n_2|2}^{(0)}(J_2|X,\tilde{X})+\mathcal{W}_{g-1,n+1|2}^{(0)}(X,J|X,\tilde{X})\Bigr)+\mathcal{R}_{g,n|2}^{(0)}(J|X,\tilde{X}),\label{FWgn2}
\end{align}
where
\begin{equation}
\mathcal{R}_{g,n|2}^{(0)}(J|X,\tilde{X})=-\left(\frac{t}{N}\right)^{n+1}\sum_{m,l\geq0}X^l\frac{(l+m+1)}{1+\delta_{m,0}}(\delta_{l+m,0}+g_{l+m+1})\frac{\partial}{\partial \xi_{m}}\frac{\partial}{\partial \Psi(\tilde{X})}\prod_{i=1}^n\frac{\partial}{\partial V(x_i)}\mathcal{F}_g^{(2)},
\end{equation}
and we used
\begin{align}
\frac{\partial}{\partial \Psi(\tilde{X})}\Psi(X)=&\frac{1}{X-\tilde{X}}+\frac{1}{2\tilde{X}},\\
-\frac{t}{N}\frac{\partial}{\partial \Psi(\tilde{X})}\mathcal{R}_{n|1}^{(1)}(J|X)=&-\frac{\mathcal{W}_{n+1}^{(0)}(\tilde{X},J|)}{X-\tilde{X}}+\mathcal{R}_{n|2}^{(0)}(J|X,\tilde{X}).
\end{align}

We abuse the terminology below and we often refer to \eqref{FWgn1} or \eqref{FWgn2} the fermionic loop equations, which should be clear from the context.

In the next section, we study recursive structures for $\mathcal{W}_{g,n|0}^{(0)}(J|)$ and $\mathcal{W}_{g,n|2}^{(0)}(J|X,\tilde{X})$ by carefully observing the bosonic and fermionic loop equations \eqref{BWgn1} and \eqref{FWgn2}.

\section{Recursion in the Ramond Sector}\label{sec:recursion}

The bosonic and fermionic loop equations for the Ramond sector are complicated. For the NS sector, \cite{BO} presented that all correlation functions can be computed by the Eynard-Orantin topological recursion in conjunction with an auxiliary Grassmann-valued polynomial equation. However, their arguments are not directly solving the loop equations, but rather they heavily depend on Becker's formula \cite{Beckers,McArthur}. With the lack of such a powerful simplification, we directly research the loop equations, and present interesting phenomena in the Ramond sector.

\subsection{Planar Limit}
From now on, we set the potentials $V(x),\Psi(x)$ to be polynomials of degree $d_V,d_{\Psi}$ respectively. Explicitly, we impose
\begin{equation}
g_0=g_1=0,\;\;\;\;g_{k\geq d_V+1}=0,\;\;\;\;g_{d_V}=\frac{1}{d_V},\;\;\;\;\xi_{l\geq d_{\Psi}+1}=0,\label{choice}
\end{equation}
where $d_V$ and $d_{\Psi}$ can be chosen independently and $g_{d_V}$ is just a convenient choice. In this setting, all $x\cdot\mathcal{P}^{(0,2)}_{g,n+1|0}(x,J|)$, $x\cdot\mathcal{Q}^{(2)}_{g,n+1|0}(x,J|)$, $\mathcal{R}^{(1,3)}_{g,n|1}(J|X)$, and $\mathcal{R}^{(0)}_{g,n|2}(J|X,\tilde{X})$ become polynomials in $x$ or $X$. In particular, $x\cdot\mathcal{P}^{(0)}_{g,n+1|0}(x,J|)$ are of degree $d_V-2+\delta_{g,0}\delta_{n,0}$, and $\mathcal{R}^{(0)}_{g,n|2}(J|X,\tilde{X})$ are of degree $d_V-1$.

\subsubsection{Spectral Curve}

Let us consider the $\xi_l$-independent bosonic loop equation \eqref{BWgn1} for $g=n=0$:
\begin{equation}
0=-V'(x)\mathcal{W}^{(0)}_{0,1|0}(x|)+\frac12\left(\mathcal{W}^{(0)}_{0,1|0}(x|)\right)^2+\mathcal{P}^{(0)}_{0,1|0}(x|).\label{W01}
\end{equation}
Equivalently, if we define
\begin{equation}
y=\mathcal{W}^{(0)}_{0,1|0}(x|)-V'(x),\label{y}
\end{equation}
then we can write \eqref{W01} as
\begin{equation}
xy^2=x\left(V'(x)\right)^2-2\hat{\mathcal{P}}^{(0)}_{0,1|0}(x|)\label{curve1}
\end{equation}
where
\begin{equation}
\hat{\mathcal{P}}^{(0)}_{0,1|0}(x|)=-\sum_{m\geq0}x^m\sum_{k\geq0}(m+k+1)(\delta_{m+k,0}+g_{m+k+1})\frac{\partial}{\partial g_k}\mathcal{F}_0^{(0)}
\end{equation}
is a polynomial in $x$ of degree $d_V-1$. 

We can prove that \eqref{curve1} is indeed a genus-zero algebraic curve. More precisely, we shall prove the following proposition:
\begin{proposition}\label{prop:curve}
The algebraic curve \eqref{curve1} is factorized as
\begin{equation}
xy^2=(x-\alpha_0(t))\prod_{i=1}^{d_V-1}(x-\alpha_i(t))^2,\label{curve}
\end{equation}
where $\{\alpha_{\mu}\}$ are formal power series in $t$ such that $\alpha_0(0)=0$ and $\alpha_i(0)\neq0$.
\end{proposition}

\begin{proof}
We proceed by a similar technique used to derive the 1-cut Brown's lemma \cite{Brown,Brown2} for Hermitian matrix models (See Lemma 3.1.1 in \cite{Eynard}). 

Thanks to \eqref{Ft} and Assumption~\ref{assumption}, the $t$-dependence of $\hat{\mathcal{P}}^{(0)}_{0,1|0}(x|)$ is
\begin{equation}
\hat{\mathcal{P}}^{(0)}_{0,1|0}(x|)=2tV'(x)+\mathcal{O}(t^2).
\end{equation}
Thus, the $t$-expansion of the right-hand-side of \eqref{curve1} is given by
\begin{equation}
\text{Pol}(x)=x\bigl(V'(x)\bigr)^2-4tV'(x)+\mathcal{O}(t^2).\label{curve2}
\end{equation}
Let $\{\alpha_0(t),\alpha_1(t),...,\alpha_{2d+1}(t)\}$ be the set of $2d+1$ roots of $\text{Pol}(x)$. If we expand one of them as a power series in $t^p$, we have
\begin{eqnarray}
\alpha(t)&=&c_0+c_pt^p+c_{2p}t^{2p}+\cdots,\nonumber\\
V'(\alpha(t))&=&V'(c_0)+V''(c_0)\left(c_pt^p+c_{2p}t^{2p}+\cdots\right)+\frac{1}{2!}V'''\left(c_pt^p+c_{2p}t^{2p}+\cdots\right)^2+\cdots\nonumber\\
\bigl(V'(\alpha(t))\bigr)^2&=&\bigl(V'(c_0)\bigr)^2+2V'(c_0)V''(c_0)c_pt^p+\left(\bigl(V''(c_0)\bigr)^2c_p^2+V'(c_0)V''(c_0)c_{2p}\right)t^{2p}+\cdots\nonumber\\
\alpha(t)\bigl(V'(\alpha(t))\bigr)^2&=&\bigl(V'(c_0)\bigr)^2c_0+\left(\bigl(V'(c_0)\bigr)^2c_p+2V'(c_0)V''(c_0)c_0c_p\right)t^p\nonumber\\
&&+\biggl(\bigl(V'(c_0)\bigr)^2c_{2p}+2V'(c_0)V''(c_0)c_p^2\nonumber\\
&&+\left(\bigl(V''(c_0)\bigr)^2c_p^2+V'(c_0)V''(c_0)c_{2p}\right)c_0\biggr)t^{2p}+\cdots,\nonumber\\
\end{eqnarray}
where $c_p\neq0$. By the assumption of $\alpha(t)$, \eqref{curve2} should vanish order by order in $t^{p}$, hence, we get
\begin{equation}
\bigl(V'(c_0)\bigr)^2c_0=0.
\end{equation}
We note that $V'(0)\neq0$.

If $c_0\neq0$, $V'(c_0)=0$ and the polynomial becomes
\begin{equation}
\text{Pol}(\alpha(t))=\left(\bigl(V''(c_0)\bigr)^2c_p^2\right)t^{2p}+\mathcal{O}(t^2,t^{3p})=0.\label{curve3}
\end{equation}
Let us assume $V''(c_0)\neq0$\footnote{If $V''(c_0)=0$, we simply expand it to higher orders, similar to the 1-cut Brown's lemma \cite{Brown,Brown2}.} Then since $c_p\neq0$ by construction, $p=1$. Also, the coefficient of the $t^2$ term in \eqref{curve3} gives a quadratic equation in $c_1$, hence they always appear as a pair $c_{1+},c_{1-}$. Then similar to the 1-cut Brown's lemma, we can show that they must be identical otherwise $W_{0,1}^{(0)}(|x)$ does not become a polynomial in $1/x$ order by order in $t$. (See Lemma 3.1.1 in \cite{Eynard}.)

If $c_0=0$, 
\begin{equation}
\text{Pol}^{2d+1}(X(t))=\left(\bigl(V'(0)\bigr)^2c_p\right)t^p-4tV'(0)+\mathcal{O}(t^2,t^{2p})=0.
\end{equation}
Then, $V'(0)\neq0$ implies that $p=1$ and
\begin{equation}
c_1=\frac{4t}{V'(0)}.
\end{equation}
This means that there is only one root if $c_0=0$, which is consistent with the argument that there are $2d_V-1$ roots and $2d_V-2$ of them are pairwise as previously discussed in the case $c_0\neq0$. This proves Proposition~\ref{prop:curve}. 
\end{proof}

\subsubsection{Parametrization and Ramification Points}

As similar to Hermitian matrix models, we now re-sum the power series \eqref{curve}, and we rather interpret the coefficients in powers of $t$ as Taylor expansions of actual functions of t. From this point of view, we think of \eqref{curve} as a $t$-dependent family of genus zero curves, and evaluate correlation functions as living on the curve.

Since \eqref{curve} is a genus-zero algebraic curve, we can parametrize $x$ and $y$ with rational functions. Let $z$ be a local coordinate on the Riemann sphere, then one such parametrization is
\begin{equation}
x(z)=\frac{\alpha_0(t)}{1-z^2},\;\;\;\;y(z)=z\prod_{i=1}^n(x(z)-\alpha_i(t)),\label{parametrization}
\end{equation}
The curve \eqref{curve} possesses a natural involution $\sigma:y(z)\mapsto-y(z)$ and $\sigma:x(z)\mapsto x(z)$, which we call the hyperelliptic involution. In terms of the local coordinate $z$, it is simply given by $\sigma(z)=-z$. $x(z)$ has a simple pole at $z=\pm1$, and a double zero at $z=\infty$. Accordingly, $dx(z)$ has a simple zero at $z=0$ and at $z=\infty$ as one can see explicitly from
\begin{equation}
dx=\frac{2\alpha_0(t)z}{(1-z^2)^2}dz.
\end{equation}
Thus, the curve has ramification points at $z=0$ and at $z=\infty$. Also, we define a differential on the curve associated with $y(z)$ as
\begin{equation}
\omega_{0,1|0}^{(0)}(z|)=\frac12y(z)dx(z).\label{w10}
\end{equation}

\begin{remark}
$y(z)$ has a simple zero at $z=0$ whereas it has a simple pole at $z=\infty$. Therefore, unlike two ramification points for supereigenvalue models in the NS sector, the ramification point at $z=0$ are the Airy-like (regular ramification) point whereas the other one at $z=\infty$ is the Bessel-like (irregular ramification) point. See \cite{Bessel} for a Bessel-type topological recursion.
\end{remark}

For simplicity, we omit the $z$-dependence of $x(z),dx(z)$ with the understanding that they are a meromorphic function and a differential on the curve \eqref{curve}.

\subsubsection{Bilinear Differentials of the Second Kind}

We now upgrade correlation functions $\mathcal{W}_{g,n|m}^{(0)}(J|K)$ to multilinear differentials defined on the spectral curve \eqref{curve} and uncover their pole structure.

To begin with, let us define a bilinear differential $\omega_{0,2|0}(z_1,z_2|)$ as
\begin{equation}
\omega_{0,2|0}^{(0)}(z_1,z_2|)=-\frac12W_{0,2|0}^{(0)}(\sigma(z_1),z_2|)dx_1dx_2,
\end{equation}
where $dx_i=dx(z_i)$ and the $\frac12$ is inserted just for convention. By this equality, we mean that the Taylor expansion near $t = 0$ of the multilinear differential on the left hand side recovers the formal series of the correlation functions on the right hand side. This concept should apply to every other differentials that we shall define below. Then, we can rewrite the $\xi_l$-independent bosonic loop equation \eqref{BWgn1} for $g=0,n=1$ in terms of $\omega_{0,2|0}^{(0)}(z_1,z_2|)$ as
\begin{equation}
\omega_{0,2}^{(0)}(z,z_1|)=\frac{dxdx_1}{2y(\sigma(z))}\left(\frac{d}{dx_1}\frac{W_{0,1|0}^{(0)}(\sigma(z)|)-W_{0,1|0}^{(0)}(z_1)|}{x-x_1}+\frac{1}{x}\frac{d}{dx_1}W^{(0)}_{0,1|0}(z_1|)+P_{0,2|0}^{(0)}(x;z_1|)\right),\label{w02}
\end{equation}
where $\sigma(z)=-z$. Then, the same argument used for Hermitian matrix models can be immediately applied to investigate the pole structure of $\omega_{0,2|0}^{(0)}(z,z_1|)$. See for example Section 2.3.2 of \cite{BO}. The only difference is that the right hand side of \eqref{w02} may have a pole at the zeros of $x$; the irregular ramification point due to the last two terms in \eqref{w02}. However, as one can see, the combination $dx/y(z)x$ has no pole at $z=\infty$. Therefore, we conclude that
\begin{equation}
\omega^{(0)}_{0,2|0}(z,z_1|)=B(z,z_1)=\frac{dzdz_1}{(z-z_1)^2},\label{B}
\end{equation}
\begin{equation}
\omega^{(0)}_{0,2|0}(z,z_1|)+\omega^{(0)}_{0,2|0}(\sigma(z),z_1|)=\frac{dxdx_1}{(x-x_1)^2}.\label{B1}
\end{equation}
$B(z,z_1)$ is sometimes called the \emph{bilinear differential of the second kind} or the \emph{Bergman kernel}. It is uniquely defined on the Riemann Sphere as explicitly given in \eqref{B}, and we would need to specify the holomorphic part for Riemann surfaces with nonzero genus.

\begin{remark}\label{alpha}
Since we use this argument repetitively, let us show why there is no pole at $x=\alpha_i$. From \eqref{w02}, if  $\omega^{(0)}_{0,2|0}(z,z_1|)$ had a pole at  $x=\alpha_i$, it would be in the form $1/(x-\alpha_i(t))$. Since $\alpha_i(t)$ is a power series in $t$ with $\alpha_i(0)\neq0$, if we Taylor-expand it around $t\rightarrow0$, it gives
\begin{equation}
\frac{1}{x-\alpha_i(t)}=\sum_{k\geq0}\frac{\left(\alpha_i(0)\right)^j}{x^{j+1}}+\mathcal{O}(t).
\end{equation}
This would contribute an infinite series in $1/x$ for any fixed power of $t$, which is in contradiction to Corollary~\ref{coro:poly}. Therefore, $\omega^{(0)}_{0,2|0}(z,z_1|)$ has no pole at those points.
\end{remark}

One can follow the argument given for Hermitian matrix models further for all other $\mathcal{W}^{(0)}_{0,n+1|0}(z,J|)$. Let us define multilinear differentials by
\begin{equation}
\omega_{0,n+1|0}^{(0)}(z,J)=\frac12\mathcal{W}^{(0)}_{0,n+1|0}(z,J|)dxdx_1\cdots dx_n,
 \end{equation}
again this equality means that after Taylor expanding the left hand side near $t = 0$ we recover the formal series of the correlation functions on the right hand side. Then, we can show that by induction in $n\geq2$,
\begin{equation}
\omega_{0,n+1|0}^{(0)}(z,J|)+\omega_{0,n+1|0}^{(0)}(\sigma(z),J|)=0,\label{involution1}
\end{equation}
and $\omega_{0,n+1|0}^{(0)}(z,J)$ have poles only at the regular ramification point $z=0$. Note that to prove these results, we used the argument in Remark~\ref{alpha}, the fact that the combination $dx/y(z)x$ has no pole at $z=\infty$, and also the fact that $\mathcal{P}^{(0)}_{0,n+1|0}(x,J|)$ behaves as $x^{d_V-3}$ in the limit $x\rightarrow\infty$. This implies that the bosonic loop equation \eqref{BWgn1} with $g=0$ and $n\geq2$ is simply written by
\begin{equation}
\omega^{(0)}_{0,n+1|0}(z,J|)=\frac{1}{2\omega_{0,1|0}^{(0)}(z|)}\sum_{J_1\cup J_2=J}^*\omega^{(0)}_{0,n_1+1|0}(z,J_1|)\omega^{(0)}_{0,n_2+1|0}(\sigma(z),J_2|)\;+\;(\text{regular at $z=0$}),
\end{equation}
where $\sum^*$ excludes the term involving $\omega_{0,1|0}^{(0)}(z|)$. Therefore, we arrive at the $g=0$ Eynard-Orantin topological recursion 
\begin{equation}
\omega^{(0)}_{0,n+1|0}(z,J|)=\underset{q\rightarrow0}{\text{Res}}\;K(z,q,\sigma(q))\sum_{J_1\cup J_2=J}^*\omega^{(0)}_{0,n_1+1|0}(q,J_1|)\omega^{(0)}_{0,n_2+1|0}(\sigma(q),J_2|),
\end{equation}
\begin{equation}
K(z,q,\sigma(q))=\frac{\int^q_{\sigma(q)}\omega_{0,2|0}^{(0)}(\cdot,z)}{2(\omega_{0,1|0}^{(0)}(q|)-\omega_{0,1|0}^{(0)}(\sigma(q)|))},\label{K}
\end{equation}
where to derive \eqref{K} we used the fact that the sum of all residues is zero on the Riemann sphere. Note that the non-recursive polynomial terms $P_{0,n+1|0}^{(0)}(x,J|)$ do not appear in the formula, which is a main benefit of the Eynard-Orantin topological recursion in terms of solving the loop equations of matrix models.

\begin{remark}
The absence of poles at $z=\infty$ is not a surprise. Poles at Bessel-like points would appear in $\omega_{g,n+1}(z,J)$ for $g\geq1$ in the framework of the Eynard-Orantin topological recursion.
\end{remark}

\subsection{Differentials From Fermionic Loop Equations}

We now study the fermionic loop equations and investigate the pole structure. This part is new compared to what is discussed in \cite{BO}, and we find a recursive structure without the simplification of Becker's formula.

\subsubsection{Supersymmetric Partner}
To begin with, let us derive a supersymmetric partner of \eqref{curve}. The fermionic loop equation \eqref{FWgn1} for $g=n=0$ can be written as an Grassmann-valued equation on the curve \eqref{curve} as
\begin{equation}
0=-\Psi(x)V'(x)+y(z)\left(x\mathcal{W}^{(1)}_{0,0|1}(|z)-\Psi(x)\right)+\mathcal{R}^{(1)}_{0,0|1}(|x). \label{fcurve1}
\end{equation}
Equivalently, one may define $\gamma(z)$ by
\begin{equation}
\gamma(z)=\mathcal{W}^{(1)}_{0,0|1}(|z)-\frac{\Psi(x)}{x},
\end{equation}
then we obtain a Grassmann-valued polynomial equation among $y,x,\gamma$ as
\begin{equation}
xy(z)\gamma(z)=\Psi(x)V'(x)-\mathcal{R}^{(1)}_{0,0|1}(|x).\label{partner}
\end{equation}
This can be thought of as a supersymmetric partner of the spectral curve \eqref{curve}. This resembles to what is discussed in \cite{BO} for the NS sector. Similar to \eqref{w10}, we define a Grassmann-valued differential
\begin{equation}
\omega_{0,0|1}^{(1)}(|z)=\frac12\gamma(z)dx
\end{equation}

\begin{remark}
In this paper, we define $\gamma(z)$ as a Grassmann-valued meromorphic function defined on the \emph{ordinal} Riemann sphere associated with the algebraic curve \eqref{curve}. We leave to future work a potential interpretation in terms of a super Riemann surface.
\end{remark}

As we will use later for the pole structure of $\mathcal{W}^{(0)}_{0,0|2}(|z,\tilde{z})$, let us look at the $\xi_0$-dependence of $\mathcal{W}^{(1)}_{0,0|1}(|z)$. If we take the derivative of \eqref{fcurve1} with respect to $\xi_0$, we have\footnote{Strictly speaking, we need to act the derivative with respect to $\xi_0$ on the fermionic loop equation \eqref{Floop} as a power series of the couplings before we fix the potentials as in \eqref{choice}. Thus, \eqref{0W^1} is nonzero even if we set $\xi_0=0$.}
\begin{equation}
\frac{\partial}{\partial\xi_0}\mathcal{W}^{(1)}_{0,0|1}(|z)dx=\frac{dx}{xy(z)}\left(y(z)+V'(x)-\frac{\partial}{\partial\xi_0}\mathcal{R}^{(1)}_{0,0|1}(|x)\right).\label{0W^1}
\end{equation}
The first term can be simply written as $dx/x$. It has no pole at the zeros of $y(z)=0$ due to the same reason as Remark~\ref{alpha}; otherwise it contradicts to Corollary~\ref{coro:poly} that it is a polynomial of $1/x$ order by order in $t$. The last term in \eqref{0W^1} is a polynomial in $x$ of degree $d_V-2$, not $d_V-1$ due to the derivative with respect to $\xi_0$, hence it does not give a pole at $x=\infty$. Therefore, this term has no contribution to the pole structure. Note that the last term on its own has poles at the zeros of $y(z)$, but they are precisely cancelled by the other terms in \eqref{0W^1}. The second term has a pair of simple poles at $z=1$, and another pole at the opposite sign at $z=\sigma(1)$. Also, it would have another pair of the same simple poles at $z=\sigma(-1),-1$. Therefore, we obtain
\begin{equation}
\frac{\partial}{\partial\xi_0}\mathcal{W}^{(1)}_{0,0|1}(|z)dx=\frac{dx}{x}+\omega^{1}_{-1}(z)=\frac{2zdz}{1-z^2}+\frac{2dz}{(z-1)(z+1)},\label{0W^1dx}
\end{equation}
where $\omega^a_b(z)$ is the fundamental differential of the third kind, that is, it is a differential that has a simple pole at $z=a$ and another pole with the opposite sign at $z=b$. Note that if $a,b$ are the poles of $x$, i.e., $a=1,b=-1$ then it obeys an additional property
\begin{equation}
\omega^1_{-1}(\sigma(z))=-\omega^1_{-1}(z).
\end{equation}
\begin{remark}
Note that the pole structure of \eqref{0W^1dx}  is independent of the degree of the curve \eqref{curve}. Therefore, \eqref{0W^1dx} holds for any choice of $V(x),\Psi(x)$.
\end{remark}

\subsubsection{Antisymmetric Bilinear Differentials}

We now turn to probe the pole structure of $\mathcal{W}^{(0)}_{0,0|2}(|z,z_1)$. We define an antisymmetric bilinear differential 
\begin{equation}
\omega_{0,0|2}^{(0)}(|z,\tilde{z})=-\frac12\mathcal{W}^{(0)}_{0,0|2}(|\sigma(z),\tilde{z})dxd\tilde{x}\label{B2}
\end{equation}
Then, if we substitute $g=n=0$ for \eqref{FWgn2}, we get
\begin{equation}
\omega_{0,0|2}^{(0)}(|z,\tilde{z})=\frac{dxd\tilde{x}}{2xy(\sigma(z))}\left(\frac{\mathcal{W}_{0,1|0}^{(0)}(\sigma(z)|)-\mathcal{W}_{0,1|0}^{(0)}(\tilde{z}|)}{x-\tilde{x}}+\frac{\mathcal{W}_{0,01|0}^{(0)}(\sigma(z)|)}{2\tilde{x}}+\mathcal{R}_{0,0|2}^{(0)}(|x,\tilde{z})\right).\label{B23}
\end{equation}
It is straightforward to see
\begin{equation}
\omega_{0,0|2}^{(0)}(|z,\tilde{z})+\omega_{0,0|2}^{(0)}(|\sigma(z),\tilde{z})=\frac{x+\tilde{x}}{2x\tilde{x}(x-\tilde{x})}dxd\tilde{x}.\label{B3}
\end{equation}

Recall that $y(z)\sim\pm V'(x)$ at $z\rightarrow\pm1$ ($x\rightarrow\infty$), and also that $\mathcal{R}_{0,0|2}^{(0)}(|z,\tilde{z})$ is a polynomial in $x$ of degree $d_V-1$ whose coefficient of $x^{dV-1}$ is precisely \eqref{0W^1dx} divided by -2. Then, the pole structure of the antisymmetric bilinear differential $\omega_{0,0|2}^{(0)}(|z,\tilde{z})$ is summarized as follows:
\begin{equation}
\omega_{0,0|2}^{(0)}(|z,\tilde{z})\sim
\begin{cases}
\frac{dz}{z-\tilde{z}}\frac{d\tilde{x}}{\tilde{x}} & z\rightarrow\tilde{z}\\
\frac{dz}{z\mp1}\left(-\frac{d\tilde{x}}{4\tilde{x}} \pm \frac14\omega^1_{-1}(\tilde{z})\right)& z\rightarrow\pm1\\
\frac{dz}{z-\sigma(\pm1)}\left(-\frac{d\tilde{x}}{4\tilde{x}} \mp \frac14\omega^1_{-1}(\tilde{z})\right)& z\rightarrow\sigma(\pm1)\\
\frac{dz}{z}\left(-\frac{d\tilde{x}}{2\tilde{x}}\right) & z\rightarrow\infty,\sigma(\infty)
\end{cases}.\label{list}
\end{equation}
Note that since $x$ has a double zero at $z=\infty$, the pole at $x\rightarrow0$ in terms of the local coordinate $z\rightarrow\infty$ is not $-d\tilde{x}/4\tilde{x}$, but $-d\tilde{x}/2\tilde{x}$. With the parametrization \eqref{parametrization}, we can write $\omega_{0,0|2}^{(0)}(|z,\tilde{z})$ in the form
\begin{equation}
\omega_{0,0|2}^{(0)}(|z,\tilde{z})=\frac{(z+\tilde{z})(1-z\tilde{z})}{(z-\tilde{z})(1-z^2)(1-\tilde{z}^2)}dzd\tilde{z}.\label{fff}
\end{equation}
One can explicitly check that \eqref{fff} reproduces all the pole structure listed in \eqref{list}. Note that after some simplification, we can write it in the following expression:
\begin{equation}
\omega_{0,0|2}^{(0)}(|z,\tilde{z})=\frac{x^2-\tilde{x}^2}{2x\tilde{x}}B(z,\tilde{z})+\frac18\alpha_0(t)\frac{x-\tilde{x}}{x\tilde{x}}\omega^1_{-1}(z)\omega^1_{-1}(\tilde{z}),\label{ff}
\end{equation}
where the $\alpha_0(t)$ is inserted in the second term so that it is independent of the choice of the curve \eqref{curve}. This expression becomes helpful to see how a supersymmetric correction appears for $g=1$ loop equations.

\begin{remark}
One way of deriving \eqref{fff} is the following. Notice that the pole structure on the list \eqref{list} is independent of the choice of the potentials $V(x),\Psi(x)$. Thus, $\omega_{0,0|2}^{(0)}(|z,\tilde{z})$ should have the same expression in terms of the local coordinate $z$ no matter what the degrees of spectral curve \eqref{curve} and its super-partner \eqref{partner} are. In particular, we can find the explicit representation of $\omega_{0,0|2}^{(0)}(|z,\tilde{z})$ by the simplest example where all $g_k=0$. In this case, $V'(x)=1$, $y=z$, and $\mathcal{R}_{0,0|2}^{(0)}(|X,\tilde{X})$ precisely corresponds to \eqref{0W^1dx} divided by $-2$. Then, we obtain \eqref{fff} from directly calculating the loop equation \eqref{B23}.
\end{remark}

We indeed have obtained all the initail data required for the recursion of all other correlation functions, namely:
\begin{itemize}
\item the algebraic curve \eqref{curve},
\item the Grassmann-valued polynomial equation \eqref{partner},
\item the symmetric bilinear differential $\omega_{0,2|0}(z_1,z_2|)$ defined in \eqref{B},
\item the antisymmetric bilinear differential $\omega_{0,0|2}(|z,\tilde{z})$ defined in \eqref{fff}.
\end{itemize}
In the following sections, we will propose a new recursive formalism that compute all other correlation functions by these initial data.

\subsubsection{Multilinear Differentials in the Planar Limit}
We next study the pole structure of $\mathcal{W}_{0,n|2}^{(0)}(J|z,\tilde{z})$. We show a few observations and computational techniques that do not appear in the context of Hermitian matrix models. It turns out that all such techniques can be repeatedly used by induction for other genus zero and higher differentials, hence we present all of them in detail below. A summary is given in Proposition~\ref{prop:g=0} at the end of the section.

Let us define for $n\geq1$
\begin{equation}
\omega_{0,n|2}^{(0)}(J|z,\tilde{z})=\frac12\mathcal{W}_{0,n|2}^{(0)}(J|z,\tilde{z})dxd\tilde{x}\prod_{i=1}^ndx_i.
\end{equation}
The fermionic loop equation \eqref{FWgn2} for $g=0,n=1$ gives the equation for $\omega_{0,1|2}^{(0)}(z_1|z,\tilde{z})$ as
\begin{align}
&-\omega_{0,1|2}^{(0)}(z_1|z,\tilde{z})\nonumber\\
&=\frac{dxd\tilde{x}dx_1}{2xy(z)}\biggl(\frac{\mathcal{W}_{0,2|0}^{(0)}(z,J|)-\mathcal{W}_{0,2|0}^{(0)}(\tilde{z},J|)}{x-\tilde{x}}+\frac{\mathcal{W}_{0,2|0}^{(0)}(z,J|)}{2\tilde{x}}\nonumber\\
&\;\;\;\;+\frac{d}{dx_1}\frac{x\mathcal{W}_{0,0|2}^{(0)}(|z,\tilde{z})-x_1\mathcal{W}_{0,0|2}^{(0)}(|z_i,\tilde{z})}{x-x_1}+x\mathcal{W}_{0,2|0}^{(0)}(z,z_1|)\mathcal{W}_{0,0|2}^{(0)}(|z,\tilde{z})+\mathcal{R}_{0,n|2}^{(0)}(z_1|x,\tilde{z})\biggr)\label{W031}\\
&=-\frac{1}{2\omega_{0,1|0}^{(0)}(z|)}\left(\omega^{(0)}_{0,2|0}(z,z_1|)\omega^{(0)}_{0,0|2}(|\sigma(z),\tilde{z})+\omega^{(0)}_{0,2|0}(\sigma(z),z_1|)\omega^{(0)}_{0,0|2}(|z,\tilde{z})\right)+\frac{F_{0,1|2}^{(0)}(z_1|x,\tilde{z})dxd\tilde{x}dx_1}{2xy(z)},\label{W032}
\end{align}
where  we used \eqref{B1} and \eqref{B3} at the second equality and $F_{0,1|2}^{(0)}(x)$ is a function of $x$ given by
\begin{equation}
F_{0,1|2}^{(0)}(z_1|x,\tilde{z})=-\frac{\mathcal{W}_{0,2|0}^{(0)}(\tilde{z},z_1|)}{x-\tilde{x}}-\frac{d}{dx_1}\frac{x_1\mathcal{W}_{0,0|2}^{(0)}(|z_1,\tilde{z})}{x-x_1}+\mathcal{R}_{0,1|2}^{(0)}(z_1|x,\tilde{z}).
\end{equation}
In particular, $F(x)$ is regular at $x=0$ and $F(x)\sim x^{d_V-1}$ in the limit $x\rightarrow\infty$. In terms of the local coordinate $z$, the loop equation shows that $\omega^{(0)}_{0,1|2}(z_1|z,\tilde{z})$ may have poles at $z=\tilde{z},\sigma(\tilde{z}),z_1,\sigma(z_1)$, and $x=0,\infty,\alpha_i$, and further the regular ramification point (the irregular ramification point corresponds to $x=0$). We will show below that it only have a pole at the regular ramification point.

First of all, we apply the same argument as Remark~\ref{alpha} to find that it has no pole at $x=\alpha_i$ thanks to Corollary~\ref{coro:poly}. Notice that the last term involving $F_{0,1|2}^{(0)}(x)$ has no pole at $x=0$ because the combination $dx/xy(z)$ has no pole there. Also, $\omega^{(0)}_{0,0|2}(|z,\tilde{z})$ has a simple pole at $x=0$, but $\omega_{0,1|0}^{(0)}(z|)$ in the determinant has a double pole there, thus this ensures that $\omega_{0,1|2}^{(0)}(z_1|z,\tilde{z})$ has no pole at $x=0$. A key point is that the last term in \eqref{W031} contributes to a simple pole at $z=\pm1$, the poles of $x$, which includes the non-recursive polynomial $\mathcal{R}_{0,1|2}^{(0)}(z_1|x,\tilde{z})$ in $x$. Thus, the recursion does not seem to work.

We can indeed still proceed the recursion with the following trick. Notice in the expression \eqref{W032} that the right hand side is odd under the hyperelliptic involution, which implies 
\begin{equation}
\omega_{0,1|2}^{(0)}(z_1|z,\tilde{z})+\omega_{0,1|2}^{(0)}(z_1|\sigma(z),\tilde{z})=0.\label{involution2}
\end{equation}
This is similar to \eqref{involution1}. It is clear from \eqref{W031} that $\omega_{0,1|2}^{(0)}(z_1|z,\tilde{z})$ has no pole at $z\rightarrow\tilde{z}$ and at $z\rightarrow z_1$. Now \eqref{involution2} shows that $\omega_{0,1|2}^{(0)}(z_1|z,\tilde{z})$ has no pole either at $z\rightarrow\sigma(\tilde{z})$ nor at $z\rightarrow \sigma(z_1)$. More importantly, since $z$ and $\tilde{z}$ should be antisymmetric by definition, $\omega_{0,1|2}^{(0)}(z_1|z,\tilde{z})$ has simple zeros at $z\rightarrow\tilde{z}$ and $\sigma(\tilde{z})$. As a consequence, if we define another multilinear differential 
\begin{equation}
\hat{\omega}_{0,1|2}^{(0)}(z_1|z,\tilde{z})=\frac{\omega_{0,1|2}^{(0)}(z_1|z,\tilde{z})}{x-\tilde{x}},
\end{equation}
then $\hat{\omega}_{0,1|2}^{(0)}(z_1|z,\tilde{z})$ is still regular at $z\rightarrow\tilde{z},\sigma(\tilde{z})$. Furthermore, the denominator practically decreases the order of poles of $\omega_{0,1|2}^{(0)}(z_1|z,\tilde{z})$ at $x\rightarrow\infty$ by 1. In conclusion, we have shown that $\hat{\omega}_{0,1|2}^{(0)}(z_1|z,\tilde{z})$ has poles only at the regular ramification point in terms of $z$, and $\tilde{z}$.

In summary, \eqref{W032} can be written in terms of $\hat{\omega}_{0,1|2}^{(0)}(z_1|z,\tilde{z})$ and $\hat{\omega}_{0,0|2}^{(0)}(|z,\tilde{z})$ as
\begin{align}
\hat{\omega}_{0,1|2}^{(0)}(z_1|z,\tilde{z})=&\frac{1}{2\omega_{0,1|0}^{(0)}(z|)}\left(\omega^{(0)}_{0,2|0}(z,z_1|)\hat{\omega}^{(0)}_{0,0|2}(|\sigma(z),\tilde{z})+\omega^{(0)}_{0,2|0}(\sigma(z),z_1|)\hat{\omega}^{(0)}_{0,0|2}(|z,\tilde{z})\right)\nonumber\\
&+\;(\text{regular at $z=0$})\label{regular}
\end{align}
where
\begin{equation}
\hat{\omega}_{0,0|2}^{(0)}(|z,\tilde{z})=\frac{\omega_{0,0|2}^{(0)}(|z,\tilde{z})}{x-\tilde{x}}.
\end{equation}
Then, we apply the recursion kernel and the residue formula to \eqref{regular}, and we obtain:
\begin{equation}
\hat{\omega}_{0,1|2}^{(0)}(z_1|z,\tilde{z})=\underset{q\rightarrow0}{\text{Res}}\;K(z,q,\sigma(q))\left(\omega^{(0)}_{0,2|0}(q,z_1|)\hat{\omega}^{(0)}_{0,0|2}(|\sigma(q),\tilde{z})+\omega^{(0)}_{0,2|0}(\sigma(q),z_1|)\hat{\omega}^{(0)}_{0,0|2}(|q,\tilde{z})\right).\label{W012}
\end{equation}

A bonus is that we can use \eqref{W012} to investigate the pole structure in terms of $z_1$. Since $\omega^{(0)}_{0,2|0}(\sigma(q),z_1|)$ can only have a double pole at $z_1=0$, we learn from \eqref{W012} that the only possible pole of $\hat{\omega}_{0,1|2}^{(0)}(z_1|z,\tilde{z})$ in terms of the variable $z_1$ is also at $z_1=0$. Furthermore, recall from \eqref{B} and \eqref{ff} that 
\begin{align}
\omega^{(0)}_{0,2|0}(\sigma(q),z_0|)&=\omega^{(0)}_{0,2|0}(q,\sigma(z_0)|),\;\;\;\;\omega^{(0)}_{0,2|0}(q,z_0|)=\omega^{(0)}_{0,2|0}(\sigma(q),\sigma(z_0)|),\\
\omega^{(0)}_{0,0|2}(|\sigma(q),\tilde{z})&=\omega^{(0)}_{0,0|2}(|q,\sigma(\tilde{z})),\;\;\;\;\;\;\;\omega^{(0)}_{0,0|2}(|q,\tilde{z})=\omega^{(0)}_{0,0|2}(|\sigma(q),\sigma(\tilde{z})).
\end{align}
Then \eqref{W012} tells us that it satisfies
\begin{equation}
\hat{\omega}_{0,1|2}^{(0)}(\sigma(z_0)|z,\tilde{z})=\hat{\omega}_{0,1|2}^{(0)}(z_0|z,\sigma(\tilde{z}))=-\hat{\omega}_{0,1|2}^{(0)}(z_0|z,\tilde{z}).\label{involution3}
\end{equation}

Let us now define for $n\geq1$,
\begin{equation}
\hat{\omega}_{0,n|2}^{(0)}(J|z,\tilde{z})=\frac{\omega_{0,n|2}^{(0)}(J|z,\tilde{z})}{x-\tilde{x}}.
\end{equation}
Repeating the same steps, we can show by induction that $\hat{\omega}_{0,n|2}^{(0)}(J|z,\tilde{z})$ satisfies
\begin{equation}
\hat{\omega}_{0,n|2}^{(0)}(J|z,\tilde{z})+\hat{\omega}_{0,n|2}^{(0)}(J|\sigma(z),\tilde{z})=0,
\end{equation}
and these differentials have poles only at the regular ramification point. 

\begin{proposition}\label{prop:g=0}
For $n\geq1$, $\hat{\omega}_{0,n|2}^{(0)}(J|z,\tilde{z})$ is obtained by
\begin{align}
\hat{\omega}_{0,n|2}^{(0)}(J|z,\tilde{z})=&\underset{q\rightarrow0}{{\rm Res}}\;K(z,q,\sigma(q))\sum_{J_1\cup J_2=J}^*\biggl(\omega^{(0)}_{0,n_1+1|0}(q,J_1|)\hat{\omega}^{(0)}_{0,n_2|2}(J_2|\sigma(q),\tilde{z})\nonumber\\
&\hspace{45mm}+\omega^{(0)}_{0,n_1+1|0}(\sigma(q),J_1|)\hat{\omega}^{(0)}_{0,n_2|2}(J_2|q,\tilde{z})\biggr),
\end{align}
where $\sum^*$ excludes the term involving $\omega_{0,1|0}^{(0)}(z|)$. Moreover, $\hat{\omega}_{0,n|2}^{(0)}(J|z,\tilde{z})$ are odd under the hyperelliptic involution $\sigma$ in terms of every variable in $(J,z,\tilde{z})$
\end{proposition}
\begin{proof}
For $\hat{\omega}_{0,1|2}^{(0)}(z_0|z,\tilde{z})$, we have already proved that the proposition and corollary are correct. Then suppose both of them hold for $\hat{\omega}_{0,n+1|2}^{(0)}(z_0,J|z,\tilde{z})$ up to some $n$. Then repeating the argument presented for $\hat{\omega}_{0,1|2}^{(0)}(z_0|z,\tilde{z})$, we can easily clarify the proposition and corollary from their loop equations.
\end{proof}

\subsection{Higher Genus: Departure from the Eynard-Orantin Topological Recursion}

All multilinear differentials $\omega^{(0)}_{0,n+1|0}(z,J|)$ are obtained by the Eynard-Orantin topological recursion with the curve \eqref{curve}. For differentials with $g\geq1$, however, this does not hold anymore and we are in need of modifying the recursion formula. Let us first summarize the procedure how we investigated the pole structures for planar differentials $\omega^{(0)}_{0,n+1|0}(z,J|)$ and $\omega^{(0)}_{0,n|2}(J|z,\tilde{z})$:

\begin{remark}[Procedure of Pole Structure Analysis]\label{rem:list}\hfill
\begin{enumerate}
\item Choose a differential $\omega^{(0)}_{0,n+1|0}(z,J|)$ or $\omega^{(0)}_{0,n|2}(J|z,\tilde{z})$ and write down the associated loop equation from either \eqref{BWgn1} or \eqref{FWgn2},
\item Confirm with Corollary~\ref{coro:poly} that the differential of interest has no pole at $x=\alpha_i$ following the argument in Remark~\ref{alpha},
\item Show that the differential is odd under the hyperelliptic involution as extension of \eqref{involution1} and \eqref{involution3}. This implies that the differential has no pole at the coinciding points of variables,
\item Clarify that the differential has no pole at $x=0$ with the care for the terms involving $\omega^{(0)}_{0,0|2}(|z,\tilde{z})$ and $dx/xy(z)$,
\item Justify that $\omega^{(0)}_{0,n+1|0}(z,J|)$ or $\hat{\omega}^{(0)}_{0,n|2}(J|z,\tilde{z})$ have no pole at $x\rightarrow\infty$ by looking at the behaviour of $\mathcal{P}_{g,n+1|0}^{(0)}(x,J|)$ and $\mathcal{R}_{g,n|2}^{(0)}(J|x,\tilde{z})$. Note that $\omega^{(0)}_{0,n|2}(J|z,\tilde{z})$ do have poles at the poles of $x$,
\item Conclude that  $\omega^{(0)}_{0,n+1|0}(z,J|)$ and $\hat{\omega}^{(0)}_{0,n|2}(J|z,\tilde{z})$ only have poles at the regular ramification point $z=0$, 
\item Obtain the recursive formulae by applying the recursion kernel $K(z,q,\sigma(q))$ and by picking the residues as in Proposition~\ref{prop:g=0},
\item Check that $\hat{\omega}^{(0)}_{0,n+1|2}(z_0,J|z,\tilde{z})$ has poles only at $z_0=0$ by the recursive formula, and also that it is odd under the hyperelliptic involution in terms of $z_0$ as in Proposition~\ref{prop:g=0}.
\item Repeat these steps inductively.
\end{enumerate}
\end{remark}

We proceed with the study of pole structures for higher genus differentials following the list above. Thus, we will only present computational key points which have not appeared in the analysis of planar differentials, and skip the rest which should be easy to follow.

\subsubsection{Supersymmetric Corrections}
It turns out that for higher genus differentials, we only need to be careful for $\omega_{1,1|0}^{(0)}(z|)$ because a supersymmetric correction appears for it. The loop equation for $\omega_{1,1|0}^{(0)}(z|)$ becomes
\begin{equation}
-\omega_{1,1|0}^{(0)}(z|)=\frac{1}{2\omega_{0,1|0}^{(0)}(z|)}\left(-\frac{1}{2}\omega_{0,2|0}^{(0)}(z,\sigma(z)|)-\frac{x}{2}\hat{\omega}_{0,0|2}^{(0)}(|z,\sigma(z))+\mathcal{P}^{(0)}_{1,1|0}(x|)\right).
\end{equation}
Corollary~\ref{coro:poly} ensures that there is no pole at $x=\alpha_i(t)$. The first two terms have poles both at $z=0,\infty$, and it is easy to show that $\mathcal{P}^{(0)}_{1,1|0}(x|)$ does not contribute to any pole. Also, $x\,\omega_{0,0|2}^{(0)}(|z,\sigma(z))$ gives no pole at $x=\infty$ with $\omega_{0,1|0}^{(0)}(z|)$ in the denominator. Unlike the genus zero differentials, however, each of the first two terms has a pole at $z=\infty$. Therefore, one might infer at first glance that the residue formula would look like
\begin{equation}
\omega_{1,1|0}^{(0)}(z|)=\underset{q\rightarrow0,\infty}{{\rm Res}}\;K(z,q,\sigma(q))\left(\frac{1}{2}\omega_{0,2|0}^{(0)}(q,\sigma(q)|)+\frac{x(z)}{2}\hat{\omega}_{0,0|2}^{(0)}(|q,\sigma(q))\right).
\end{equation}
In particular, if we substitute \eqref{ff} in the above equation, we get
\begin{equation}
\omega_{1,1|0}^{(0)}(z|)=\underset{q\rightarrow0,\infty}{{\rm Res}}\;K(z,q,\sigma(q))\left(\omega_{0,2|0}^{(0)}(z,\sigma(z)|)-\frac{\alpha_0(t)}{16x}\left(\omega^1_{-1}(z)\right)^2\right).\label{W111}
\end{equation}
The first term is precisely the same as the Eynard-Orantin topological recursion and the second term can be thought of as a supersymmetric correction. Note that in the NS sector the second term vanishes due to Becker's formula, but this is not the case any more in the Ramond sector.

Now let us probe carefully the pole at $z=\infty$. The $\omega_{0,1|0}^{(0)}(z|)$ and $\omega_{0,2|0}^{(0)}(z,\sigma(z)|)$ at $p=1/z\rightarrow0$ would behave as
\begin{equation}
\frac{1}{\omega_{0,1|0}^{(0)}(p|)}\sim\frac{1}{dp},\;\;\;\;\omega_{0,2|0}^{(0)}(p,\sigma(p)|)\sim\frac{dpdp}{p^2}.
\end{equation}
Thus, their product has a double pole at $p\rightarrow0$, and as a result, $\omega_{1,1|0}^{(0)}(z|)$ would have a double pole at $z\rightarrow\infty$ ($p\rightarrow0$) \emph{if there were no second term in} \eqref{W111}. This is what happens if we do the Eynard-Orantin topological recursion at any Bessel-like (irregular) ramification point.

However, it turns out that the second term in \eqref{W111} precisely cancels the pole at $z=\infty$ coming from the first term. In fact if we explicitly use the parametrization \eqref{parametrization}, we can show that
\begin{equation}
\omega_{0,2|0}^{(0)}(z,\sigma(z)|)-\frac{\alpha_0(t)}{16x}\left(\omega^1_{-1}(z)\right)^2=-\frac{dzdz}{4z^2}+\frac{dzdz}{4(z^2-1)}=\frac{1}{4z^2(1-z^2)}dzdz.
\end{equation}
In particular, in the limit $p=1/z\rightarrow0$, it becomes regular
\begin{equation}
\frac{1}{2}\omega_{0,2|0}^{(0)}(z,\sigma(z)|)+\frac{x(z)}{2}\hat{\omega}_{0,0|2}^{(0)}(|z,\sigma(z))\sim dpdp.
\end{equation}
Therefore, $\omega_{1,1|0}^{(0)}(z|)$ has no pole at $z=\infty$. Note that it still has a pole at $z=0$, and the second term in \eqref{W111} also contribute to the pole. In conclusion, the recursive formula for $\omega_{1,1|0}^{(0)}(z|)$ is
\begin{equation}
\omega_{1,1|0}^{(0)}(z|)=\underset{q\rightarrow0}{{\rm Res}}\;K(z,q,\sigma(q))\left(\omega_{0,2|0}^{(0)}(z,\sigma(z)|)-\frac{\alpha_0(t)}{16x}\left(\omega^1_{-1}(z)\right)^2\right).
\end{equation}

\subsubsection{Formula for Higher Genus}

For any other $\omega_{1,n+1|0}^{(0)}(z,J|)$, the pole structure can be investigated by mostly following the procedure listed in Remark~\ref{rem:list}. The terms that do not appear for $g=0$ differentials are
\begin{equation}
\frac{1}{2}\frac{\omega_{0,n+2|0}^{(0)}(z,\sigma(z),J|)}{\omega_{0,1|0}^{(0)}(u|)}+\frac{x}{2}\frac{\hat{\omega}_{0,n|2}^{(0)}(J|z,\sigma(z))}{\omega_{0,1|0}^{(0)}(u|)}.\label{term1}
\end{equation}
We have shown in the previous section that all $\omega_{0,n+2|0}^{(0)}(z,†{z},J|)$ and $\hat{\omega}_{0,n|2}^{(0)}(J|z,\tilde{z})$ for $n\geq1$ have poles only at $z=0$. This implies that \eqref{term1} does not have poles at $z=\infty$. Thus, by induction in $n$, we can show that all $\omega_{1,n+1|0}^{(0)}(z,J|)$ have poles only at $z=0$  following  Remark~\ref{rem:list}. We thus apply the recursion kernel and get
\begin{align}
\omega_{1,n+1|0}^{(0)}(z,J|)=&\underset{q\rightarrow0}{\text{Res}}\;K(z,q,\sigma(q))\biggl(\frac12\omega_{0,n+2|0}^{(0)}(q,\sigma(q),J|)+\frac{x(q)}2\hat{\omega}_{0,n|2}^{(0)}(J|q,\sigma(q))\nonumber\\
&+\sum_{\substack{g_1+g_2=1\\J_1\cup J_2=J}}^*\omega^{(0)}_{g_1,n_1+1|0}(q,J_1|)\omega^{(0)}_{g_2,n_2+1|0}(\sigma(q),J_2|).\biggr)
\end{align}

To compute $\hat{\omega}_{1,n|2}^{(0)}(J|z,\sigma(z))$, one might wonder that terms with $\hat{\omega}_{0,2|0}^{(0)}(z,\sigma(z)|)$ would contribute to a simple pole at $z\rightarrow\pm1$ $(x\rightarrow\infty)$ and also at $z=\infty$ $(x\rightarrow0)$. However, since $\omega_{0,1|0}^{(0)}(z|)$ in the denominator gives a simple zero at $z=\pm1$ and a double zero at $z=\infty$, such terms do not contribute to poles at $z=\pm1,\infty$. Also, Proposition~\ref{prop:g=0} implies that $\hat{\omega}_{0,n+1|2}^{(0)}(\sigma(q),J|q,\tilde{z})$ has a pole only at $q=0$. Thus, if we follow Remark~\ref{rem:list}, we can show from the fermionic loop equation \eqref{FWgn2} that $\hat{\omega}_{1,n|2}^{(0)}(J|z,\sigma(z))$ has poles only at $z=0$, and the recursive formula can be written in terms of the residue at $z=0$ as
\begin{align}
\hat{\omega}_{1,n|2}^{(0)}(J|z,\tilde{z})=&\underset{q\rightarrow0}{\text{Res}}\;K(z,q,\sigma(q))\biggl(\hat{\omega}_{0,n+1|2}^{(0)}(\sigma(q),J|q,\tilde{z})\nonumber\\
&+\sum_{\substack{g_1+g_2=1\\J_1\cup J_2=J}}^*\Bigl(\omega^{(0)}_{g_1,n_1+1|0}(q,J_1|)\hat{\omega}^{(0)}_{g_2,n_2|2}(J|\sigma(q),\tilde{z})\nonumber\\
&\hspace{30mm}+\omega^{(0)}_{g_1,n_1+1|0}(\sigma(q),J_1|)\hat{\omega}^{(0)}_{g_2,n_2|2}(J|q,\tilde{z})\Bigr)\biggr).
\end{align}

The discussion above can be easily extended to differentials for all $g\geq1$ without any trouble by induction in $2g+n$. Thus, one can indeed obtain the recursive formula for all  $\omega_{g,n+1|0}^{(0)}(z,J|)$ and  $\hat{\omega}_{g,n|2}^{(0)}(J|z,\tilde{z})$. We summarize the formulae below:
\begin{align}
\omega_{g,n+1|0}^{(0)}(z,J|)=&\underset{q\rightarrow0}{\text{Res}}\;K(z,q,\sigma(q))\biggl(\frac12\omega_{g-1,n+2|0}^{(0)}(q,\sigma(q),J|)+\frac{x(q)}2\hat{\omega}_{g-1,n|2}^{(0)}(J|q,\sigma(q))\label{R1}\nonumber\\
&+\sum_{\substack{g_1+g_2=g\\J_1\cup J_2=J}}^*\omega^{(0)}_{g_1,n_1+1|0}(q,J_1|)\omega^{(0)}_{g_2,n_2+1|0}(\sigma(q),J_2|)\biggr),\\
\hat{\omega}_{g,n|2}^{(0)}(J|z,\tilde{z})=&\underset{q\rightarrow0}{\text{Res}}\;K(z,q,\sigma(q))\biggl(\hat{\omega}_{g-1,n+1|2}^{(0)}(\sigma(q),J|q,\tilde{z})\nonumber\\
&+\sum_{\substack{g_1+g_2=g\\J_1\cup J_2=J}}^*\Bigl(\omega^{(0)}_{g_1,n_1+1|0}(q,J_1|)\hat{\omega}^{(0)}_{g_2,n_2|2}(J|\sigma(q),\tilde{z})\nonumber\\
&\hspace{20mm}+\omega^{(0)}_{g_1,n_1+1|0}(\sigma(q),J_1|)\hat{\omega}^{(0)}_{g_2,n_2|2}(J|q,\tilde{z})\Bigr)\biggr)\label{R2}
\end{align}
Also, they are all odd under the hyperelliptic involution in terms of any variable as the extension of Proposition~\ref{prop:g=0}.

\begin{remark}
\eqref{R1} is different from the Eynard-Orantin topological recursion, hence this can be thought of as its generalization. A more abstract framework will be introduced in \cite{BO2}.
\end{remark}

\section{$\xi_l$-Dependent Differentials}

We have shown that all $\xi_l$-independent correlation functions can be recursively computed. However, they are defined not as power series in $1/x$, but as meromorphic differentials on the curve \eqref{curve}. As a result, the formulae \eqref{W1formula} and \eqref{W2formula} cannot be immediately applied. They are written as the residue in terms of $X$, hence we need to rewrite them in a way that we can pull them back to the curve similar to what is discussed in \cite{BO}.

\subsection{Supersymmetric Partner: Revisited}

 \cite{BO} presented a set of equations to compute the $\xi_l$-dependent differentials utilizing the supersymmetric curve for the NS sector (see Theorem 4.1 in \cite{BO}). An interesting question is whether analogous equations for the Ramond sector also exist. That is, we would like to find a Grassmann-valued meromorphic function $\hat{\gamma}$ such that $y,x,\hat{\gamma}$ satisfy a polynomial equation, and furthermore it computes the $\xi_l$-dependent differentials as follows
\begin{align}
\omega_{g,n|1}^{(1)}(J|z)&\overset{?}{=}\frac12\underset{\tilde{z}\rightarrow0}{{\rm Res}}\hat{\gamma}(\tilde{z})\omega^{(0)}_{g,n|2}(J|\tilde{z},z),\label{R5}\\
\omega_{g,n|0}^{(2)}(J|)&\overset{?}{=}\frac14\underset{\tilde{z}\rightarrow0}{{\rm Res}}\hat{\gamma}(\tilde{z})\hat{\gamma}(z)\omega^{(0)}_{g,n|2}(J|z,\tilde{z}).\label{R6}
\end{align}

A candidate of such $\hat{\gamma}(z)$ is $\gamma(z)$ or $x\gamma(z)$, defined in \eqref{partner} as the super partner of the curve \eqref{curve}. However, it turns out that these do not serve quite well. This is because in terms of a power series in $x$ and $1/x$, $x\gamma(x)$ obeys
\begin{equation}
\underset{x\rightarrow\infty}{{\rm Res}}x\gamma(x)\frac{\partial}{\partial\Psi(x)}dx=\underset{x\rightarrow\infty}{{\rm Res}}\left(x\mathcal{W}_{0,0|1}^{(1)}(|x)-\Psi(x)\right)\frac{\partial}{\partial\Psi(x)}dx=-\sum_{k\geq0}\xi_l\frac{\partial}{\partial\xi_l}-\eta\frac{\partial}{\partial\xi_0},\label{candidate}
\end{equation}
\begin{equation}
\eta=-\frac12\left(\xi_0-\frac{\partial\mathcal{F}^{(2)}_0}{\partial\xi_0}\right).\label{eta}
\end{equation}
The first term on the right hand side of \eqref{candidate} is proportional to the identity operator for the $\xi_l$-dependent functions and differentials, which is what we would like to have as we disucssed in \eqref{I0}. Our task is to modify $x\gamma(z)$ to get rid of the term proportional to Grassmann constant $\eta$.

\subsubsection{A Conservative Modification}

Let us consider non-radical conditions that we impose on such a modification. First, we require that the modified Grassmann-valued meromorphic function on the curve \eqref{curve},
\begin{equation}
\hat{\gamma}(z)=x\gamma(z)-\zeta(z),\label{hat}
\end{equation}
still satisfies a Grassmann-valued polynomial equation with $y$ and $x$. This means that a simple choice for $\zeta(z)$ is in the form
\begin{equation}
\zeta(z)=\frac{\text{Pol}^{(1)}(x)}{y(z)},
\end{equation}
where $\text{Pol}^{(1)}(x)$ is some Grassmann-valued polynomial of $x$. In particular, $\zeta(z)$ is odd under the hyperelliptic involution. Furthermore, we would like to have formulae in terms of the residue only at $z\rightarrow0$ as in \eqref{R5} and \eqref{R6}, which suggests that $\zeta(z)$ has poles only at $z=0$. These conditions restrict $\zeta(z)$ to be in the form
\begin{equation}
\zeta(z)=\frac{2\zeta_0}{y(z)}\prod_{i=1}^{d_V-1}(x-\alpha_i)=\frac{2\zeta_0}{z},
\end{equation}
where $\zeta_0$ is a Grassmann constant and we have
\begin{equation}
y(z)\hat{\gamma}(z)=V'(x)\Psi(x)-\mathcal{R}^{(1)}_{0,0|1}(|x)+2\zeta_0\prod_{i=1}^{d_V-1}(x-\alpha_i).\label{hat(poly)}
\end{equation}
In fact, if we take the residue at $x=\infty$ as a power series in $1/x$ by substituting \eqref{y} into $y$,
\begin{equation}
-\underset{x\rightarrow\infty}{{\rm Res}}\frac{2\zeta_0}{y(x)}\prod_{i=1}^{d_V-1}(x-\alpha_i)\frac{\partial}{\partial\Psi(x)}=\zeta_0\frac{\partial}{\partial\xi_0}\label{zeta}
\end{equation}
This implies that if we choose $\zeta_0=\eta$, then we have
\begin{equation}
\underset{x\rightarrow\infty}{{\rm Res}}\hat{\gamma}(x)\frac{\partial}{\partial\Psi(x)}dx=-\sum_{k\geq0}\xi_l\frac{\partial}{\partial\xi_l}.\label{I}
\end{equation}
This is proportional to the identity operator for $\xi_l$-dependent differentials.

Let us apply \eqref{I} to the $\xi_l$-dependent differentials. First, let us consider for $2g+n\geq1$, 
\begin{equation}
\underset{x(\tilde{z})\rightarrow\infty}{{\rm Res}}\hat{\gamma}(\tilde{z})\omega^{(0)}_{g,n|2}(J|\tilde{z},z)=-\omega^{(1)}_{g,n|1}(J|z).
\end{equation}
This is equivalent to \eqref{W1formula}. Note that since both $\hat{\gamma}(\tilde{z})$ and $\omega^{(0)}_{g,n|2}(J|\tilde{z},z)$ are odd under the hyperbolic involution, we can pull it back to the curve and evaluate the residues at $z=\pm1$:
\begin{equation}
\frac12\left(\underset{\tilde{z}\rightarrow1}{{\rm Res}}+\underset{\tilde{z}\rightarrow-1}{{\rm Res}}\right)\hat{\gamma}(\tilde{z})\omega^{(0)}_{g,n|2}(J|\tilde{z},z)=-\omega^{(1)}_{g,n|1}(J|z).
\end{equation}
Finally, recall that the sum of all residues vanishes on any compact Riemann surface and both $\hat{\gamma}(\tilde{z})$ and $\omega^{(0)}_{g,n|2}(J|\tilde{z},z)$ have poles only at $z=0,\pm1$. Thus, we get
\begin{equation}
\omega^{(1)}_{g,n|1}(J|z)=\frac12\underset{\tilde{z}\rightarrow0}{{\rm Res}}\hat{\gamma}(\tilde{z})\omega^{(0)}_{g,n|2}(J|\tilde{z},z).
\end{equation}
Similarly, we have
\begin{equation}
\omega^{(2)}_{g,n|1}(J|z)=\frac18\underset{z\rightarrow0}{{\rm Res}}\hat{\gamma}(z)\omega^{(1)}_{g,n|1}(J|z).
\end{equation}

\subsubsection{Relation Between $\gamma(z)$ and $\hat{\gamma}(z)$}

From the fermionic loop equation perspective, $\gamma(z)$ seems a fundamental meromorphic function as a super partner of $y(z)$. On the other hand, $\hat{\gamma}(z)$ is crucial to compute all the $\xi_l$-dependent differentials. Even though their relation is explicitly given in \eqref{eta} and \eqref{hat}, it would be helpful to find a more geometric understanding between them.

We would like to maintain the two conditions on $\hat{\gamma}(z)$, namely, (i) $y\hat{\gamma}$ gives a Grassmann-valued polynomial, and (ii) $\hat{\gamma}(z)$ has poles only at $z=0$. Therefore, with these constraints, the remaining question is how to geometrically obtain the constant $\zeta_0$. Then, motivated by \eqref{partner} and \eqref{B3}, we consider the following equation as a power series in $x,1/x$
\begin{equation}
\underset{\tilde{x}\rightarrow\infty}{{\rm Res}}\hat{\gamma}(\tilde{x})\left(\frac12\mathcal{W}_{0,0|2}^{(0)}(|x,\tilde{x})-\frac{x+\tilde{x}}{4x\tilde{x}(x-\tilde{x})}\right)dxd\tilde{x}=\frac12\left(\mathcal{W}_{0,0|1}^{(1)}(|x)-\frac{\Psi(x)}{x}\right)dx,
\end{equation}
where we used
\begin{equation}
\underset{\tilde{x}\rightarrow\infty}{{\rm Res}}\tilde{x}^l\frac{x+\tilde{x}}{2x\tilde{x}(x-\tilde{x})}d\tilde{x}=-\underset{\tilde{x}\rightarrow\infty}{{\rm Res}}\tilde{x}^l\frac{1}{x}\left(-\frac{1}{2\tilde{x}}+\sum_{k\geq0}\frac{x^k}{\tilde{x}^{k+1}}\right)d\tilde{x}=\frac{x^{l-1}}{1+\delta_{l,0}}.
\end{equation}
If we manipulate this result with \eqref{partner} and \eqref{B3}, we end up with
\begin{equation}
-\underset{x(z)\rightarrow\infty}{{\rm Res}}\hat{\gamma}(z)\left(\omega_{0,0|2}^{(0)}(|z,\tilde{z})-\omega_{0,0|2}^{(0)}(|\sigma(z),\tilde{z})\right)=\gamma(\tilde{z})d\tilde{x}.\label{hat2}
\end{equation}

Recall from \eqref{list} that $\omega_{0,0|2}^{(0)}(|z,\tilde{z})$ has simple poles at $z=\pm1,\infty,\tilde{z},\sigma(\tilde{z})$. Also, $\hat{\gamma}(z)$ has a simple pole at $z=0$, accordingly it has a simple zero at $z=\infty$. Then, since the integrand is even under the hyperelliptic involution, we pull the residue formula \eqref{hat2} back to the curve and evaluate it as:
\begin{align}
\gamma(\tilde{z})d\tilde{x}&=-\frac12\left(\underset{z\rightarrow1}{{\rm Res}}+\underset{z\rightarrow-1}{{\rm Res}}\right)\hat{\gamma}(z)\left(\omega_{0,0|2}^{(0)}(|z,\tilde{z})-\omega_{0,0|2}^{(0)}(|\sigma(z),\tilde{z})\right)\nonumber\\
&=\frac12\left(\underset{z\rightarrow \tilde{z}}{{\rm Res}}+\underset{z\rightarrow\sigma(\tilde{z})}{{\rm Res}}+\underset{z\rightarrow0}{{\rm Res}}\right)\hat{\gamma}(z)\left(\omega_{0,0|2}^{(0)}(|z,\tilde{z})-\omega_{0,0|2}^{(0)}(|\sigma(z),\tilde{z})\right)\nonumber\\
&=\hat{\gamma}(\tilde{z})\frac{d\tilde{x}}{\tilde{x}}+\frac12\underset{z\rightarrow0}{{\rm Res}}\;\hat{\gamma}(z)\left(\omega_{0,0|2}^{(0)}(|z,\tilde{z})-\omega_{0,0|2}^{(0)}(|\sigma(z),\tilde{z})\right)\nonumber\\
&=\gamma(\tilde{z})d\tilde{x}-\zeta(\tilde{z})\frac{d\tilde{x}}{\tilde{x}}+\frac12\underset{z\rightarrow0}{{\rm Res}}\;\hat{\gamma}(z)\left(\omega_{0,0|2}^{(0)}(|z,\tilde{z})-\omega_{0,0|2}^{(0)}(|\sigma(z),\tilde{z})\right).
\end{align}
Since $\zeta(z)$ is defined as \eqref{zeta}, we simply have
\begin{equation}
-\frac12\underset{z\rightarrow0}{{\rm Res}}\zeta(z)\left(\omega_{0,0|2}^{(0)}(|z,\tilde{z})-\omega_{0,0|2}^{(0)}(|\sigma(z),\tilde{z})\right)=\frac12\zeta(\tilde{z})\frac{d\tilde{x}}{\tilde{x}}.
\end{equation}
Thus, we arrive at the following equation
\begin{equation}
\zeta(\tilde{z})\frac{d\tilde{x}}{\tilde{x}}=\underset{z\rightarrow0}{{\rm Res}}\;x(z)\gamma(z)\left(\omega_{0,0|2}^{(0)}(|z,\tilde{z})-\omega_{0,0|2}^{(0)}(|\sigma(z),\tilde{z})\right).
\end{equation}

\subsection{Recursion in the Ramond Sector}

Let us summarize what we have shown so far. We first recall the initial data for the recursion.

\begin{definition}[Initial data for the recursion]\label{def:scurve}
Given a choice of polynomial potentials $V(x),\Psi(x)$ with a parameter $t$, we define a quintuple $\mathcal{S}=\{x,y,\gamma,\omega_{0,2|0}^{(0)},\omega_{0,0|2}^{(0)}\}$ where:
\begin{itemize}
\item $x,y$ defines an algebraic curve of genus zero as in \eqref{curve},
\item $x,y,\gamma$ satisfies a Grassmann-valued polynomial relation as in \eqref{partner},
\item We parametrize $x$ and $y$ as \eqref{parametrization},
\item $\omega_{0,2|0}^{(0)}$ is the bilinear differential of the second kind, that is,
\begin{equation}
\omega_{0,2|0}^{(0)}(z_1,z_2|)=B(z_1,z_2)=\frac{dz_1dz_2}{(z_1-z_2)^2}.
\end{equation}
\item $\omega_{0,0|2}^{(0)}$ is given by
\begin{align}
\omega_{0,0|2}^{(0)}(|z,\tilde{z})=\frac{(z+\tilde{z})(1-z\tilde{z})}{(z-\tilde{z})(1-z^2)(1-\tilde{z}^2)}dzd\tilde{z}.
\end{align}
\end{itemize}
\end{definition}

Starting with $\mathcal{S}$, let us define an infinite sequence of multilinear differentials on the algebraic curve \eqref{curve} in terms of correlation functions of a supereigenvalue model in the Ramond sector by
\begin{align}
\omega_{0,1|0}^{(0)}(z|)&=\frac12y(z)dx=\frac12\left(\mathcal{W}^{(0)}_{0,1|0}(z|)-V'(x)\right)dx,\\
\omega_{0,0|1}^{(0)}(|z)&=\frac12\gamma(z)dx=\frac12\left(\mathcal{W}^{(1)}_{0,0|1}(|z)-\Psi(x)\right)dx,\\
\omega_{0,2|0}^{(0)}(z_1,z_2|)&=\frac12\left(\mathcal{W}^{(0)}_{0,2|0}(z_1,z_2|)-\frac{1}{(x_1-x_2)^2}\right)dx_1dx_2,\\
\omega_{0,0|2}^{(0)}(|z,\tilde{z})&=\frac12\left(\mathcal{W}^{(0)}_{0,0|2}(|z,\tilde{z})-\frac{x+\tilde{x}}{x\tilde{x}(x-\tilde{x})}\right)dxd\tilde{x},\\
\omega_{g,n|m}^{(a)}(J|K)&=\frac12\mathcal{W}^{(a)}_{g,n|m}(J|K)\prod_{i=1}^{n+m}dx_i\;\;\;\;(2g+n+m+a\geq3).
\end{align}
We mean by this equality that the differentials reproduces the right hand side as power series in $1/x$ after Taylor expansion around $t=0$. The superscript $(a)$ denotes the dependence of fermionic couplings, hence those with nonzero $(a)$ are Grassmann-valued. In addition, we define the recursion kernel $K(z,q,\sigma(q))$, a set of multilinear differentials $\hat{\omega}^{(0)}_{g,n|2}(J|z,\tilde{z})$, and a Grassmann-valued meromorphic function $\hat{\gamma}$ as follows
\begin{align}
K(z,q,\sigma(q))&=\frac12\frac{\int_{\sigma(q)}^q\omega_{0,2|0}^{(0)}(z,\cdot|)}{\omega_{0,1|0}^{(0)}(z|)-\omega_{0,1|0}^{(0)}(\sigma(z)|)},\label{K1}\\
\hat{\omega}^{(0)}_{g,n|2}(J|z,\tilde{z})&=\frac{\omega^{(0)}_{g,n|2}(J|z,\tilde{z})}{x-\tilde{x}},\\
\hat{\gamma}(z)&=x\gamma(z)-\frac{x}{dx}\underset{\tilde{z}\rightarrow0}{{\rm Res}}\;x\,\gamma(\tilde{z})\left(\omega_{0,0|2}^{(0)}(|\tilde{z},z)-\omega_{0,0|2}^{(0)}(|\sigma(\tilde{z}),z)\right).\label{hat1}
\end{align}
Note that $y\cdot\hat{\gamma}$ also satisfies a polynomial equation as in \eqref{hat(poly)}.

\begin{theorem}\label{thm:main}
All the differentials $\omega_{g,n|m}^{(a)}(J|K)$ with $2g+n+m+a\geq3$ are obtained by the following set of formulae:
\begin{align}
\omega_{g,n+1|0}^{(0)}(z,J|)=&\underset{q\rightarrow0}{{\rm Res}}\;K(z,q,\sigma(q))\biggl(\frac12\omega_{g-1,n+2|0}^{(0)}(q,\sigma(q),J|)+\frac{x(q)}2\hat{\omega}_{g-1,n|2}^{(0)}(J|q,\sigma(q))\label{S1}\nonumber\\
&+\sum_{\substack{g_1+g_2=g\\J_1\cup J_2=J}}^*\omega^{(0)}_{g_1,n_1+1|0}(q,J_1|)\omega^{(0)}_{g_2,n_2+1|0}(\sigma(q),J_2|)\biggr),\\
\hat{\omega}_{g,n|2}^{(0)}(J|z,\tilde{z})=&\underset{q\rightarrow0}{{\rm Res}}\;K(z,q,\sigma(q))\biggl(\hat{\omega}_{g-1,n+1|2}^{(0)}(\sigma(q),J|q,\tilde{z})\nonumber\\
&\;\;\;\;\;\;\;\;+\sum_{\substack{g_1+g_2=g\\J_1\cup J_2=J}}^*\biggl(\omega^{(0)}_{g_1,n_1+1|0}(q,J_1|)\hat{\omega}^{(0)}_{g_2,n_2|2}(J|\sigma(q),\tilde{z})\nonumber\\&\;\;\;\;\;\;\;\;\;\;\;\;\;\;\;\;\;\;\;\;\;\;\;\;+\omega^{(0)}_{g_1,n_1+1|0}(\sigma(q),J_1|)\hat{\omega}^{(0)}_{g_2,n_2|2}(J|q,\tilde{z})\biggr)\biggr),\label{S2}\\
\omega_{g,n|2}^{(0)}(J|z,\tilde{z})&=(x-\tilde{x})\,\hat{\omega}_{g,n|2}^{(0)}(J|z,\tilde{z}),\\
\omega^{(1)}_{g,n|1}(J|z)&=\frac12\underset{\tilde{z}\rightarrow0}{{\rm Res}}\,\hat{\gamma}(\tilde{z})\omega^{(0)}_{g,n|2}(J|\tilde{z},z),\label{S3}\\
\omega^{(2)}_{g,n|0}(J|)&=\frac14\underset{z\rightarrow0}{{\rm Res}}\,\hat{\gamma}(z)\omega^{(1)}_{g,n|1}(J|z).\label{S4}
\end{align}
Moreover, all of them are odd under the hyperelliptic involution $\sigma$ in terms of every variable.
\end{theorem}

\section{Discussion and Future Work}\label{sec:discussion}

In the paper, we investigated various properties of supereigenvalue models in the Ramond sector. First, we proved that the free energy depends on Grassmann couplings only up to quadratic order. Then, with the assumption that the partition function and the free energy enjoy the $1/N$ expansion (Assumption~\ref{assumption}), we derived bosonic and fermionic loop equations. We probed the pole structures of correlation functions by analyzing the loop equations, and in particular we found the associated genus-zero algebraic curve (Proposition~\ref{prop:curve}) whose ramification points are one Airy-like, and one Bessel-like. Finally starting with a quintuple $\mathcal{S}$ (Definition~\ref{def:scurve}), we obtained a set of recursive formulae that solves the loop equations of supereigenvalue models in the Ramond sector (Theorem~\ref{thm:main}). Note that no differentials have a pole at the Bessel-like ramification point, unlike the Eynard-Orantin topological recursion, which can be thought of as a supersymmetric correction. There are several interesting open questions; hence, let us summarize them in the same order as the outline of the paper.

\begin{description}
\item[Truncation] It is an interesting consequence that both supereigenvalue models in the NS sector and Ramond sector hold this feature as proved in Proposition~\ref{prop:trun}. The technique used in the proof closely follows the one used for the NS sector originally shown in \cite{McArthur}. This requires careful permutations of indices of integration variables as well as the properties of the Pfaffian, but it has nothing to do with any supersymmetric algebra in the computation. Therefore, one may ask:
\begin{itemize}
\item \textit{Can we prove the truncation from a different point of view such that it conceptually relates to super Virasoro constraints?}
\item \textit{What if we consider multi-cut or multi-covered supereigenvalue models (if such models exist)? Does the number of cuts or coverings relate to the degree of the truncation? Following a similar spirit, is it possible to construct supereigenvalue models with more than one supersymmetries ($\mathcal{N}=2$ or $\mathcal{N}=4$)? How does the truncation of their free energy appear?}
\end{itemize}

\hfill
\item[$\boldsymbol{1/N}$ Assumption] This is a crucial assumption we made in the paper, and some propositions and theorems hold only subject to this assumption. Note that strictly speaking, all we need for Theorem~\ref{thm:main} is the $1/N$ expansion for the free energy, but not necessarily for the partition function. Our justification for this assumption was initially only the fact that the models are analogous to those in the NS sector. However, since it has been proven with the assumption that all correlation functions can be recursively computed, it also enhances the plausibility of the assumption. The remaining questions towards this direction are:
\begin{itemize}
\item \textit{Can we rigorously prove the $1/N$ expansion of the partition function and the free energy?}
\item \textit{If possible, can we further relate these models to some matrix models, and can we understand combinatoric perspectives?}
\end{itemize}

\hfill
\item[Correlation Functions] Recall that we defined the fermionic loop insertion operator  \eqref{Finsertion} as a power series in $1/X$. This is because we wanted to evaluate all correlation functions as meromorphic multilinear differentials on $\mathcal{S}$. In particular, $X^{\frac12}(z)$ is not a meromorphic function, which does not work well with many known techniques to derive topological recursion. 

However, what if we naively defined the fermionic loop insertion operator as the expansion of $1/X^{n+\frac12}$ for $n\geq0$? This modification would explain the somewhat bizarre $x(q)$ factor in \eqref{S1}. More explicitly, let us consider non-meromorphic multi-valued differentials $\Omega^{(1)}_{g,n|1}(J|z)$, $\Omega^{(0)}_{g,n|2}(J|z,\tilde{z})$, and $\hat{\Omega}^{(0)}_{g,n|2}(J|z,\tilde{z})$ by
\begin{align}
\Omega^{(1)}_{g,n|1}(J|z)&=x^{\frac12}\omega^{(0)}_{g,n|1}(J|z),\label{Omega1}\\
\Omega^{(0)}_{g,n|2}(J|z,\tilde{z})&=x^{\frac12}\tilde{x}^{\frac12}\omega^{(0)}_{g,n|2}(J|z,\tilde{z}),\label{Omega2}\\
\hat{\Omega}^{(0)}_{g,n|2}(J|z,\tilde{z})&=x^{\frac12}\tilde{x}^{\frac12}\hat{\omega}^{(0)}_{g,n|2}(J|z,\tilde{z}).\label{Omega3}
\end{align}
Then, \eqref{hat1} and \eqref{S1} would be schematically rewritten in a more symmetric form:
\begin{align}
\omega_{g,n+1|0}^{(0)}(z,J|)=&\underset{q\rightarrow0}{{\rm Res}}\;K(z,q,\sigma(q))\biggl(\frac12\omega_{g-1,n+2|0}^{(0)}(q,\sigma(q),J|)+\frac12\hat{\Omega}_{g-1,n|2}^{(0)}(J|q,\sigma(q))\nonumber\\
&+\sum_{\substack{g_1+g_2=g\\J_1\cup J_2=J}}^*\omega^{(0)}_{g_1,n_1+1|0}(q,J_1|)\omega^{(0)}_{g_2,n_2+1|0}(\sigma(q),J_2|)\biggr),\label{S5}\\
\hat{\Gamma}(z)&=\Gamma(z)-\frac{1}{2dx}\underset{\tilde{z}\rightarrow0}{{\rm Res}}\,\Gamma(\tilde{z})\left(\Omega_{0,0|2}^{(0)}(|\tilde{z},z)-\Omega_{0,0|2}^{(0)}(|\sigma(\tilde{z}),z)\right),\label{S6}
\end{align}
where $\Gamma(z)=x^{\frac12}\,\gamma(z), \hat{\Gamma}(z)=x^{-\frac12}\,\hat{\gamma}(z)$. We note that \eqref{S3} and \eqref{S4} still precisely hold by replacing  $\hat{\gamma}(z)$ and $\omega_{g,n|m}^{(a)}(J|K)$ with $\hat{\Gamma}(z)$ and $\Omega_{g,n|m}^{(a)}(J|K)$. In contrast, the recursion for $\hat{\Omega}_{g,n|2}^{(0)}(J|z,\tilde{z})$ will be modified from \eqref{S2} to get the correct non-meromorphic differentials, which we shall discuss shortly.

Besides the question of how to make sense of the multi-valued function $x^{\frac12}$ and differentials \eqref{Omega1}-\eqref{Omega3}, one may also ask:
\begin{itemize}
\item \textit{Since fermionic fields are normally interpreted as sections on a spin bundle in physics, would it make the formalism simpler if we define half-order differentials instead of standard differentials for $\omega^{(0)}_{g,n|2}(J|z,\tilde{z})$?}
\end{itemize}

\hfill
\item[Antisymmetric Bilinear Differentials]
Relating to the previous discussion, let us ignore the ambiguity of $x^{\frac12}$ for now and naively consider the antisymmetric bilinear differential $\Omega_{0,0|2}^{(0)}(|z,\tilde{z})$ as in \eqref{Omega2}. First of all, in terms of the local coordinate $z$, $\hat{\Omega}_{0,0|2}^{(0)}(|z,\tilde{z})$ is written as
\begin{align}
\hat{\Omega}_{0,0|2}^{(0)}(|z,\tilde{z})=\frac{1-z\tilde{z}}{(z-\tilde{z})^2(1-z^2)^{\frac12}(1-\tilde{z}^2)^{\frac12}}dzd\tilde{z}.\label{ffff}
\end{align}
Note that in the limit $z\rightarrow\tilde{z}$, we have
\begin{equation}
\hat{\Omega}_{0,0|2}^{(0)}(|z,\tilde{z})=\frac{dzd\tilde{z}}{(z-\tilde{z})^2}+\text{reg. at }z\rightarrow\tilde{z}.
\end{equation}
Thus, it behaves analogously to the Bergman kernel. 

Let us point out one more aspect which motivates us to consider a formalism with $x^{\frac12}$. Notice that
\begin{equation}
\hat{\Omega}_{0,0|2}^{(0)}(|z,v)=\frac{dz}{(1-z^2)^{\frac12}}dv\frac{d}{dv}\frac{(1-v^2)^{\frac12}}{z-v}.
\end{equation}
By using this relation, one can explicitly check the following
\begin{equation}
\int_{\sigma(q)}^q\hat{\Omega}_{0,0|2}^{(0)}(|z,\cdot)=\frac{x(z)^{\frac12}}{x(q)^{\frac12}}\int_{\sigma(q)}^q{\omega}_{0,2|0}^{(0)}(z,\cdot|).
\end{equation}
The right hand side is precisely what appears if we rewrite the recursion \eqref{S2} in terms of $\hat{\Omega}_{g,n|2}^{(0)}(J|z,\tilde{z})$. Therefore, let us naively define another recursion kernel $\hat{K}(z,q,\sigma(q))$ by
\begin{equation}
\hat{K}(z,q,\sigma(q))=\frac12\frac{\int_{\sigma(q)}^q\hat{\Omega}_{0,0|2}^{(0)}(|z,\cdot)}{\omega_{0,1|0}^{(0)}(z|)-\omega_{0,1|0}^{(0)}(\sigma(z)|)},
\end{equation}
then the recursive formula for $\hat{\Omega}_{g,n|2}^{(0)}(J|z,\tilde{z})$ would be schematically given by
\begin{align}
\hat{\Omega}_{g,n|2}^{(0)}(J|z,\tilde{z})=&\underset{q\rightarrow0}{{\rm Res}}\;\hat{K}(z,q,\sigma(q))\biggl(\hat{\Omega}_{g-1,n+1|2}^{(0)}(\sigma(q),J|q,\tilde{z})\nonumber\\
&\;\;\;\;\;\;\;\;+\sum_{\substack{g_1+g_2=g\\J_1\cup J_2=J}}^*\biggl(\omega^{(0)}_{g_1,n_1+1|0}(q,J_1|)\hat{\Omega}^{(0)}_{g_2,n_2|2}(J|\sigma(q),\tilde{z})\nonumber\\&\;\;\;\;\;\;\;\;\;\;\;\;\;\;\;\;\;\;\;\;\;\;\;\;+\omega^{(0)}_{g_1,n_1+1|0}(\sigma(q),J_1|)\hat{\Omega}^{(0)}_{g_2,n_2|2}(J|q,\tilde{z})\biggr)\biggr).\label{S7}
\end{align}
It is worth mentioning that the product of $\hat{K}(z,q,\sigma(q))$ and $\hat{\Omega}_{g,n|2}^{(0)}(J|q,\tilde{z})$ is well-defined meromorphic differentials in terms of $q$ so that the residue formula can still make sense. 

A bonus is that we could view supereigenvalue models in the NS sector and the Ramond sector from one single recursive formalism. Namely, we take \eqref{S5}, \eqref{S6}, and \eqref{S7} as the universal recursive formulae, and what we choose is the initial condition. For example, if we choose $\hat{\Omega}_{0,0|2}^{(0)}(|z_1,z_2)=B(z_1,z_2)$ in $\mathcal{S}$ as the initial condition instead of \eqref{ffff}, and if we apply these new formulae \eqref{S5} and \eqref{S7}, then we recursively find that $\hat{\Omega}_{g,n|2}^{(0)}(J,|z,\tilde{z})=\omega_{g,n+2|0}^{(0)}(z,\tilde{z},J|)$, which is a consequence of Becker's formula. In particular, \eqref{S5} reduces down to the Eynard-Orantin topological recursion as discussed in \cite{BO}. Furthermore, if we consider an analogous equation of \eqref{S6} in the NS sector, we find that the second term vanishes and it gives $\Gamma^{NS}(z)=\hat{\Gamma}^{NS}(z)$. Therefore, we exactly recover Theorem 4.1 in \cite{BO} for the NS sector from the same recursive formulae with a different initial condition.

\begin{itemize}
\item \textit{Is there any formalism that naturally unifies the recursion for both the NS and Ramond sector?}
\item \textit{What are geometric characteristics of $\hat{\Omega}_{0,0|2}^{(0)}(|z,\tilde{z})$?}
\end{itemize}

\hfill

All of the above arguments suggest that the definitions \eqref{Omega1}-\eqref{Omega3} seem to serve better for unification of the NS and Ramond sector. One way of dealing with $x^{\frac12}$ is to re-evaluate Theorem~\ref{thm:main} on a different algebraic curve which is defined by $y$ and $s$ where $s^2=x$, so that $s``="x^{\frac12}$ becomes a meromorphic function on the new curve. Accordingly, we find another natural involution $\rho$ that didn't appear in the previous formalism, namely $\rho$ sends $s\mapsto-s$ and $\Gamma\mapsto-\Gamma$ but keeps $y$ and $x=s^2$ invariant. Moreover, if we denote by $u$ a local coordinate of the new curve and by $B(u_1,u_2)$ the Bergman kernel defined on the new curve, then $\omega_{0,2|0}^{(0)}(u_1,u_2|)$ and $\hat{\Omega}_{0,0|2}^{(0)}(|u_1,u_2)$ are written in a nice form as:
\begin{align}
\omega_{0,2|0}^{(0)}(u_1,u_2|)=&B(u_1,u_2)+B(\rho(u_1),u_2),\\
\hat{\Omega}_{0,0|2}^{(0)}(|u_1,u_2)=&B(u_1,u_2)-B(\rho(u_1),u_2).
\end{align}

Motivated from this observation, we are currently constructing an abstract recursive framework which answers some of the questions raised above for unification. Since, however, this is quite different from the scope of this paper, we leave further discussions to a paper \cite{BO2}.

\hfill
\item[Super Partner Polynomial Equation] Let us come back to Theorem~\ref{thm:main} in terms of meromorphic differentials on $\mathcal{S}$. The super partner \eqref{partner} of the algebraic curve \eqref{curve} is utilized for computing the $\xi_l$-dependent differentials, though it is not as simple as the one for the NS sector. We note that the recursive formulae \eqref{S1} and \eqref{S2} hold without knowing any data of the super partner \eqref{partner}. Therefore, the $\xi_l$-independent differentials are the same for any choice of $\Psi(x)$. This leads us to a question what the super partner equation and $\xi_l$-dependent differentials are computing if we apply this formalism to examples beyond supereigenvalue models.

Also, our formalism is based on an ordinal algebraic curve of genus zero, and consider Grassmann-valued differentials on the curve. However, since the super partner curve is discovered in both the NS sector and the R sector, one may try to reformulate them with a super Riemann surface as the underlying geometry. For example, \cite{Tori} discusses that supertori are algebraic curves. We do not have a clear understanding in this aspect at the moment:

\begin{itemize}
\item \textit{What is the meaning of the super partner polynomial equation beyond supereigenvalue models?}
\item \textit{Can we reformulate the recursive equations with the concept of super Riemann surfaces?}
\end{itemize}

\hfill
\item[Super Airy Structures]

Recently, \cite{BCHORS} proposed another supersymmetric topological recursion from an algebraic point of view, what the authors call \emph{super Airy structures}, as a generalization of \cite{KS,ABCD}. Their starting point is a super Virasoro-\emph{like} constraint acting on a partition function, and they prove that there exists a unique partition function for a given set of super-Virasoro operators, and such a partition function can be recursively computed. Note that their super Virasoro-\emph{like} constraint is essentially different from that for supereigenvalue models introduced in Section~\ref{sec:SEM}. In the context of super Airy structures, supersymmetric corrections also appear only for genus one or higher, in a similar way to \eqref{S1}. It deserves further investigation to see whether their recursion is related to the recursion presented in this paper. \cite{BO2} will discuss more about this.

Even without supersymmetry, it is an interesting question whether we can consider the standard Airy structure that captures the $\xi_l$-independent recursive formulae \eqref{S1} and \eqref{S2}. However, note that the differentials obtained from \eqref{S2} are not fully symmetric among thier variables, hence we need a more generalized formalism than Airy structures. Potentially geometric recursion recently proposed in \cite{ABO} can be useful for realizing this recursion from a categorical perspective:

\begin{itemize}
\item \textit{Is it possible to interpret the recursion for the Ramond sector as an example of super Airy structures?}
\item \textit{Is geometric recursion helpful to understand the recursive formalism presented in Theorem~\ref{thm:main}?}
\end{itemize}

\end{description}

Besides the open questions above, there are a few more that we can pursue. First of all, it is interesting to see how the recursive structure appears in the Ramond-R type models introduced in \cite{C}. Since the explicit definition is given, similar computational techniques described in this paper might be applicable to these models. Also, \cite{C1,C2,C} discussed supereigenvalue models in length in the context of super quantum curves. It is certainly worth investigating more on those aspects.

A bit differently, Theorem~\ref{thm:main} shows that two systems of multilinear differentials mix up in their recursion. Since it is known that differentials obtained by the Eynard-Orantin topological recursion have some enumerative interpretation, we suspect that the differentials from Theorem~\ref{thm:main} can be also realized in terms of curve counting. Interesting discussions of behaviours of higher genus super Riemann surfaces and their moduli spaces can be found in \cite{Donagi,W1,W2,W3,W4}. It remains to be seen further investigation to see whether these arguments have relations to the new recursion shown in this paper (or perhaps the more abstract one in \cite{BO2}). 

Finally, \cite{Stanford:2019vob} recently considered a supersymmetric recursion for super JT gravity. Then a natural question is whether supereigenvalue models can realize 2d supergravity, presumably after taking some limit as analogous to the double scaling limit in Hermitian matrix models.

\newpage
\appendix

\section{Derivation of \eqref{Bloop}}\label{sec:derivation}

We derive the bosonic and fermionic loop equation in this section. The computational technique shown here closely follows Appendix B in \cite{BO}. For simplicity, we redefine $g_1+T\rightarrow g_1$ in this section. One can simply shift the explicit dependence of $g_1$ in the loop equation by $g_1\rightarrow g_1+T$ to recover the $T$ dependence.

\subsection{Bosonic Loop Equation}
As shown in Appendix B in \cite{BO}, the contribution from the first line of \eqref{L} is
\begin{equation}
-\frac{N}{t}V'(x)\mathcal{W}^{(0,2)}_{1|0}(x|)+\frac{1}{2}\mathcal{W}^{(0,2)}_{2|0}(x,x|)+\frac{1}{2}\mathcal{W}^{(0,2)}_{1|0}(x|)^2+\mathcal{P}^{(0,2)}_{1|0}(x|)
\end{equation}
where
\begin{equation}
\mathcal{P}^{(0,2)}_{1|0}(x|)=-\sum_{m\geq0}x^{m-1}\sum_{k\geq0}(m+k+1)g_{m+k+1}\frac{\partial}{\partial g_k}\mathcal{F}^{(0,2)}.
\end{equation}
To evaluate the second line of \eqref{L}, notice that we can manipulate
\begin{align}
\sum_{k,n\geq0}\frac{1}{x^{n+2}}\left(\frac{n}{2}+k\right)\xi_{k}\partial_{\xi_{n+k}}=&\frac12 \sum_{k,n\geq0}\left(-\frac{d}{dx}\frac{1}{x^{n+1}}+\frac{2k-1}{x^{n+2}}\right)\xi_k\partial_{\xi_{n+k}}\nonumber\\
=&\frac12 \sum_{m\geq0}\sum_{k=0}^m\left(-\frac{d}{dx}\frac{x^k}{x^{m+1}}+\frac{(2k-1)x^{k-1}}{x^{m+1}}\right)\xi_k\frac{\partial}{\partial \xi_{m}}\nonumber\\
=&\frac12 \left(\sum_{m,k\geq0}-\sum_{m\geq0}\sum_{k\geq m+1}\right)\left(-\frac{d}{dx}\frac{x^k}{x^{m+1}}+\frac{(2k-1)x^{k-1}}{x^{m+1}}\right)\xi_k\frac{\partial}{\partial \xi_{m}}\nonumber\\
=&\frac12\frac{d}{dx}\Psi(x)\left(\frac{\partial}{\partial\Psi(x)}-\frac{1}{2x}\frac{\partial}{\partial \xi_0}\right)\nonumber\\
&+\frac12\left(2\Psi'(x)-\frac{\Psi(x)}{x}\right)\left(-\frac{\partial}{\partial\Psi(x)}+\frac{1}{2x}\frac{\partial}{\partial \xi_0}\right)\nonumber\\
&-\sum_{l,m\geq0}\frac{l+2m+1}{2}x^{l-1}\xi_{l+m+1}\frac{\partial}{\partial \xi_{m}}\nonumber\\
=&\frac12\left(\Psi(x)\frac{d}{dx}\frac{\partial}{\partial\Psi(x)}-\Psi'(x)\frac{\partial}{\partial\Psi(x)}+\frac{\Psi(x)}{x}\frac{\partial}{\partial\Psi(x)}\right)\nonumber\\
&-\frac{1}{2}\sum_{l,m\geq0}\frac{l+2m+1}{1+\delta_{m,0}}x^{l-1}\xi_{l+m+1}\frac{\partial}{\partial \xi_{m}}.
\end{align}
Also, we have
\begin{align}
\frac{\partial}{\partial \Psi(x)}\frac{d}{dx}\frac{\partial}{\partial \Psi(x)}=&\frac{1}{2}\sum_{k\geq0}\left(\frac{k+1}{x^{k+3}}\frac{\partial}{\partial \xi_0}\frac{\partial}{\partial \xi_k}+\frac{1}{x^{k+3}}\frac{\partial}{\partial \xi_k}\frac{\partial}{\partial \xi_0}\right)-\sum_{k,l\geq0}\frac{l+1}{x^{k+l+3}}\frac{\partial}{\partial \xi_k}\frac{\partial}{\partial \xi_l}\nonumber\\
=&\frac{1}{2}\sum_{k\geq0}\frac{k}{x^{k+3}}\frac{\partial}{\partial \xi_0}\frac{\partial}{\partial \xi_l}-\sum_{n\geq0}\sum_{l=0}^{n}\frac{n-l+1}{x^{n+3}}\frac{\partial}{\partial \xi_l}\frac{\partial}{\partial \xi_{n-l}}\nonumber\\
=&\frac{2}{x}\sum_{n\geq0}\frac{1}{x^{n+2}}\left(\frac{n}{4}\frac{\partial}{\partial \xi_0}\frac{\partial}{\partial \xi_n}-\frac{1}{2}\sum_{l=0}^{n}\left(\frac{n}{2}-l\right)\frac{\partial}{\partial \xi_l}\frac{\partial}{\partial \xi_{n-l}}\right).
\end{align}
Putting all of them together, we arrive at the bosonic loop equation \eqref{Bloop}.

\subsection{Fermionic Loop Equation}
Similarly, we can have
\begin{align}
&\sum_{m,k\geq0}\left(\frac{\xi_k}{X^{m+1}}\frac{\partial}{\partial g_{r+k}}+\frac{kg_k}{X^{m+1}}\frac{\partial}{\partial \xi_{m+k}}\right)\nonumber\\
&=\sum_{m,k\geq0}\left(\frac{\xi_kx^{k}}{X^{m+k+1}}\frac{\partial}{\partial g_{m+k}}+\frac{kg_kx^k}{X^{m+k+1}}\frac{\partial}{\partial \xi_{m+k}}\right)\nonumber\\
&=\sum_{n\geq0}\sum_{k=0}^n\left(\frac{\xi_kX^{k}}{X^{n+1}}\frac{\partial}{\partial g_{n}}+\frac{kg_kX^k}{X^{n+1}}\frac{\partial}{\partial \xi_{n}}\right)\nonumber\\
&=\left(\sum_{n,k\geq0}-\sum_{n\geq0}\sum_{k\geq n+1}\right)\left(\frac{\xi_kX^{k}}{X^{n+1}}\frac{\partial}{\partial g_{n}}+\frac{kg_kX^k}{X^{n+1}}\frac{\partial}{\partial \xi_{n}}\right)\nonumber\\
&=-\Psi(x)\frac{\partial}{\partial V(X)}-XV'(X)\left(\frac{\partial}{\partial\Psi(X)}-\frac{1}{2X}\frac{\partial}{\partial \xi_0}\right)\nonumber\\
&\;\;\;\;-\sum_{n,l\geq0}X^l\left(\xi_{l+n+1}\frac{\partial}{\partial g_{n}}+(l+n+1)g_{l+n+1}\frac{\partial}{\partial \xi_{n}}\right).
\end{align}
Also,
\begin{align}
\frac{\partial}{\partial \Psi(X)}\frac{\partial}{\partial V(X)}&=\sum_{k,l\geq0}\frac{1}{X^{k+l+2}}\frac{\partial}{\partial \xi_l}\frac{\partial}{\partial g_k}-\frac12\sum_{k\geq0}\frac{1}{X^{k+2}}\frac{\partial}{\partial g_k}\frac{\partial}{\partial \xi_{0}}\nonumber\\
&=\frac{1}{X}\sum_{m\geq0}\frac{1}{X^{m+1}}\left(\sum_{k=0}^m\frac{\partial}{\partial \xi_k}\frac{\partial}{\partial g_{m-k}}-\frac12\frac{\partial}{\partial g_m}\frac{\partial}{\partial \xi_{0}}\right)
\end{align}
Thus, \eqref{Floop} indeed holds.

\end{document}